\documentclass[10pt,journal,twocolumn,final]{IEEEtran}

\usepackage{eqnarray,amssymb,amsmath,amsthm}
\usepackage{mathrsfs}
\usepackage[utf8]{inputenc}
\usepackage[T1]{fontenc}
\usepackage{graphicx}
\usepackage{epsfig}
\usepackage[english]{babel}
\usepackage{subfigure}
\usepackage{epstopdf}
\usepackage{import}
\usepackage{cases}
\usepackage{mathtools}
\usepackage{color, colortbl}
\usepackage{cite}
\usepackage{algorithm}
\usepackage{algpseudocode}
\usepackage{bbm}
\usepackage{longtable}
\usepackage{array}
\newtheorem{lemma}{Lemma}
\newtheorem{definition}{Definition}

\newtheorem{theorem}{Theorem}
\usepackage{mathtools}



\begin{document}

\title{Dynamic Task Offloading and Resource Allocation for Ultra-Reliable Low-Latency Edge Computing}

\author{
\IEEEauthorblockN{Chen-Feng Liu,~\IEEEmembership{Student Member,~IEEE,} 
Mehdi Bennis,~\IEEEmembership{Senior Member,~IEEE,} \\M\'{e}rouane Debbah,~\IEEEmembership{Fellow,~IEEE,} and H. Vincent Poor,~\IEEEmembership{Fellow,~IEEE}
\\\emph{(Invited Paper)}
}


\thanks{This work was supported  in part  by the Academy of Finland project CARMA, in part by the Academy of Finland project MISSION, in part by the Academy of Finland project SMARTER, in part  by the INFOTECH project NOOR, in part by the Nokia Bell-Labs project FOGGY, in part by the Nokia Foundation, and in part by the U.S. National Science Foundation under Grant CNS-1702808. This paper was presented in part at the IEEE Global Communications Conference Workshops, Singapore, December 2017 \cite{GC17}.}
\thanks{C.-F. Liu and M. Bennis are with the Centre for Wireless Communications, University of Oulu, 90014 Oulu, Finland (e-mail: chen-feng.liu@oulu.fi; mehdi.bennis@oulu.fi).}
\thanks{M. Debbah is with the Large Networks and System Group, CentraleSup\'{e}lec, Universit\'{e} Paris–Saclay, 91192 Gif-sur-Yvette, France, and also with the Mathematical and Algorithmic Sciences Laboratory, Huawei France Research and Development, 92100 Paris, France (e-mail: merouane.debbah@huawei.com).}
\thanks{H. V. Poor is with the Department of Electrical Engineering, Princeton University, Princeton, NJ 08544, USA (e-mail: poor@princeton.edu).}
}

\maketitle

\begin{abstract}
To overcome devices' limitations in performing computation-intense applications, mobile edge computing (MEC) enables users to offload tasks to proximal MEC servers for faster task computation. However, current MEC system design is based on average-based metrics, which fails to account for the ultra-reliable low-latency requirements in mission-critical applications. To tackle this, this paper proposes a new system design, where probabilistic and statistical constraints are imposed on task queue lengths, by applying \emph{extreme value theory}. The aim is to minimize users' power consumption while trading off the allocated resources for local computation and task offloading. Due to wireless channel dynamics, users are re-associated to MEC servers in order to offload tasks using higher rates or accessing proximal servers. In this regard, a user-server association policy is proposed, taking into account the channel quality as well as the servers' computation capabilities and workloads. By marrying tools from Lyapunov optimization and matching theory, a two-timescale mechanism is proposed, where a user-server association is solved in the long timescale while a dynamic task offloading and resource allocation policy is executed in the short timescale. Simulation results corroborate the effectiveness of the proposed approach by guaranteeing highly-reliable task computation and lower delay performance, compared to several baselines.
\end{abstract}

\section{Introduction}\label{Sec: Intro}

\IEEEPARstart{M}{otivated} by the surging traffic demands spurred by online video and Internet-of-things (IoT) applications, including machine type and mission-critical communication (e.g., augmented/virtual reality (AR/VR) and drones),
mobile edge computing (MEC)/fog computing  are emerging technologies that distribute computations, communication, control, and storage at the network edge \cite{FogIoT,MungChiangFog,FogSurveyAccess,fog_survey,Fog_compare18}.
When executing the computation-intensive applications at mobile devices, the performance and user's quality of experience are significantly affected by the device's limited computation capability.
Additionally, intensive computations are energy-consuming which severely shortens the lifetime of  battery-limited  devices.  To address the computation and energy issues, mobile devices can wirelessly offload  their tasks  to proximal MEC servers. On the other hand, offloading  tasks incurs additional latency which cannot be overlooked and should be taken into account in the system design. Hence, the energy-delay tradeoff has received significant attention and has been studied in various MEC systems \cite{JSAC_queue,Kim18,TaskSplitting,FogPower,Mao_TWC17,EE_stable_GC17,Xu17,Sun17,KB_TWC18,Avg_delay_ICC15, Gilsoo_ICC17,FanLett18,FanJIoT18,FanCog18,Delay_estimation, Zhang_IoTJ17}.

\subsection{Related Work}\label{Sec: related work}
In \cite{JSAC_queue}, Kwak {\it et al.}~focused on an energy  minimization problem for local  computation and task offloading in a single-user MEC system.
The authors further studied a multi-user system, which takes into account both the energy cost and monetary cost of task offloading \cite{Kim18}. Therein, the cost-delay tradeoff was investigated in terms of competition and cooperation among users and offloading service provider.
Additionally,  the work  \cite{TaskSplitting} considered the single-user system and  assumed that the mobile device is endowed with a multi-core central process unit (CPU) to compute different applications simultaneously. In order to stabilize all task queues at the mobile device and MEC server, the dynamic task offloading and resource allocation policies were proposed by utilizing Lyapunov stochastic optimization in \cite{JSAC_queue,Kim18,TaskSplitting}. 
Assuming that the MEC server is equipped with multiple CPU cores to compute different users' offloaded tasks in parallel, Mao {\it et al.}~\cite{FogPower} studied a multi-user task offloading and bandwidth allocation problem. 
Subject to the stability of task queues, the energy-delay tradeoff was investigated using the  Lyapunov framework.
Extending the problem of \cite{FogPower}, the authors further took into account  the server's power consumption and  resource allocation in the system analysis \cite{Mao_TWC17}.
In  \cite{EE_stable_GC17}, a wireless powered MEC network was considered in which multiple users, without fixed energy supply, are wirelessly powered by a power beacon to carry out local computation and task offloading.
Taking into account the causality of the harvested energy, this work  \cite{EE_stable_GC17} aimed at maximizing energy efficiency subject to the stability of users' task queues.  Therein, the tradeoff between  energy efficiency and  average execution delay was analyzed by stochastic optimization.
Xu {\it et al.}~studied another energy harvesting MEC scenario, in which the edge servers are mainly powered by solar or wind energy, whereas the cloud server has a constant grid power \cite{Xu17}. Aiming at minimizing the long-term expected cost which incorporates the end-to-end delay and operational cost, the authors proposed a reinforcement learning-based resource provisioning and workload offloading (to the cloud) to edge servers.
Besides the transmission and computation delays, the work \cite{Sun17} took into account the cost (in terms of delay) of  handover  and  computation migration, due to user mobility, in an ultra-dense network. Taking into the long-term available energy constraint, an online energy-aware base station association and handover algorithm was proposed to minimize the  average end-to-end delay  by incorporating Lyapunov optimization and multi-armed bandit theory  \cite{Sun17}.
Ko {\it et al.}~\cite{KB_TWC18} analyzed the average latency performance, including communication delay and computation delay, of a large-scale spatially random MEC network.
Furthermore, an upper and a lower bound   \cite{KB_TWC18} on the average computation delay were derived by applying stochastic geometry and queuing theory.
A hybrid cloud-fog architecture was considered in \cite{Avg_delay_ICC15}. The delay-tolerable computation workloads, requested by the end users, are dispatched from the fog devices to the cloud servers when delay-sensitive workloads are computed at the fog devices. The studied problem was cast as a network-wide power minimization subject to an average delay requirement \cite{Avg_delay_ICC15}.
Focusing on the cloud-fog architecture, Lee {\it et al.}~\cite{Gilsoo_ICC17} studied a scenario in which a fog node distributes the offloaded tasks to the connected fog nodes and a remote cloud server for cooperative computation. To  address the uncertainty of the arrival of neighboring fog nodes, an online fog network formation algorithm was proposed such that  the maximal average latency among different computation nodes is minimized \cite{Gilsoo_ICC17}.
Considering a hierarchical cloudlet architecture,  Fan and Ansari  \cite{FanLett18} proposed a workload allocation (among different cloudlet tiers) and computational resource allocation approach in order to minimize the average response time of a task request.
The authors further focused on an edge computing-based IoT network in which each user equipment (UE) can run several IoT applications  \cite{FanJIoT18}. Therein, the objective was to minimize the average response time subject to the delay requirements of different applications.
In \cite{FanCog18}, a distributed workload balancing scheme was proposed for fog computing-empowered IoT networks. Based on the broadcast information of fog nodes' estimated traffic and computation loads, each IoT device locally chooses the associated fog node in order to reduce the average latency of its data flow.
In addition to the task uploading and computation phases, the work \cite{Delay_estimation} also accounted for the delay in the downlink phase, where the computed tasks are fed back to the users.
The objective was to minimize a cost function of the estimated average delays of the three phases.
The authors in \cite{Zhang_IoTJ17} studied a software-defined fog network, where the data service subscribers (DSSs) purchase the fog nodes' computation resources via the data service operators.
Modeling the average latency using queuing theory in the DSS's utility, a Stackelberg game and  a many-to-many matching game were incorporated to allocate fog nodes' resources to the DSSs  \cite{Zhang_IoTJ17}.

\subsection{Our Contribution}

\begin{table*}[t!]
\caption{Summary of Notations}\label{Tab: notations}
\centering
 \begin{tabular}{|m{1.1cm}|m{4cm}||m{1.1cm}|m{4cm}||m{1.1cm}|m{4cm}|}
\hline
{\bf Notation} &{\bf Definition}&{\bf Notation}&{\bf Definition}&{\bf Notation}&{\bf Definition}\\
\hline
  \hline
$A_i$  &  UE $i$'s task arrivals &  $A_i^{\rm L}$  & Split tasks for local computation   &  $A_{i}^{\rm O}$  &  Split tasks for offloading
\\  \hline
$A_{\rm unit}$& Unit task   &  $\tilde{A}_{i}^{\rm L}$ & Moving time-averaged split tasks for local computation &  $\tilde{A}_{i}^{\rm O}$ & Moving time-averaged split tasks for offloading
\\  \hline
 $d^{\rm L}_{i}$&Queue length bound for $Q_i^{\rm L}$& $d^{\rm O}_{i}$ &Queue length bound for $Q_i^{\rm O}$ & $d_{ji}$ & Queue length bound for $Z_{ji}$
 \\  \hline
    $f_i$&UE $i$'s CPU-cycle frequency &   $f_{ji}$&MEC server $j$'s allocated CPU-cycle frequency for UE $i$ &$f_{j}^{\max}$&MEC server $j$'s computation capability per CPU core
    \\  \hline
    $\mathbf{f}$ & Network-wide CPU-cycle frequency  vector & $\mathbf{f}_j$& MEC server $j$'s allocated CPU-cycle frequency  vector  &     $h_{ij}$&Channel gain between UE $i$ to MEC server $j$  
    \\\hline
   $I_{ij}$ & Aggregate interference to $h_{ij}$ & $L_i$&UE $i$'s required processing density  &    $\mathcal{L}$ & Lyapunov function   
   \\\hline
    $n$&Time frame index & $N_0$&Power spectral density of AWGN & $N_j$&MEC server $j$'s CPU core number   
    \\\hline
       $P_i$&UE $i$'s transmit power    & $P_i^{\max}$&UE $i$'s power budget  &  $\bar{P}_{i}^{\rm C}$&UE $i$'s long-term time-averaged
computation power   
  \\\hline
   $\bar{P}_{i}^{\rm T}$&UE $i$'s long-term time-averaged transmit power & $\mathbf{P}$ & Network-wide transmit power  vector & 	$\hat{\Pr}(I_{ij})$&Estimated distribution of $I_{ij}$	
    \\\hline
     $Q_i^{\rm L}$&UE $i$'s local-computation queue length & $Q_i^{\rm O}$&UE $i$'s task-offloading queue length    &   $Q_i^{\rm L, (X)}$&Virtual queue length for \eqref{Eq: GPD-Loc-mean}   
     \\     \hline
      $Q_i^{\rm L, (Y)}$ & Virtual queue length for \eqref{Eq: GPD-Loc-var}  &  $Q_i^{\rm O, (X)}$&Virtual queue length for \eqref{Eq: GPD-Off-mean}   &  $Q_i^{\rm O, (Y)}$&Virtual queue length  for  \eqref{Eq: GPD-Off-var}
   \\   \hline
     $Q_{ji}^{\rm (X)}$ & Virtual queue length for  \eqref{Eq: GPD-Server-mean} 
 &     $Q_{ji}^{\rm (Y)}$&Virtual queue length for  \eqref{Eq: GPD-Server-var} &  $Q_i^{\rm L, (Q)}$&Virtual queue length for  \eqref{Eq: remodel_Violation-Loc-Prob} 
      \\   \hline
       $Q_i^{\rm O, (Q)}$ &Virtual queue length for \eqref{Eq: remodel_Violation-Off-Prob} &      $Q_{ji}^{\rm  (Z)}$&Virtual queue length for \eqref{Eq: remodel_Violation-Server-Prob}   &      $\mathbf{Q}$ & Combined queue vector  
          \\   \hline
  $R_{ij}$&Transmission rate from UE $i$ to MEC server $j$  
      & $\tilde{R}_{ij}$&Moving time-averaged transmission rate from UE $i$ to MEC server $j$  &     $R^{\max}_{ij}$&Maximum offloading rate from UE $i$ to MEC server $j$ 
      \\\hline
      $\mathcal{S}$&Set of MEC servers & $S$&Number of MEC servers    &   $t$&Time slot index  
      \\\hline
$T_0$&Time frame length & $\mathcal{T}$&Time frame  &   $\mathcal{U}$ &Set of UEs
\\  \hline
  $U$&Number of UEs & $V$&Lyapunov optimization parameter  &    $W$&MEC server's dedicated bandwidth    
  \\\hline
  $X_{i}^{\rm L}$&Excess value of $Q_i^{\rm L}$ & $\bar{X}_{i}^{\rm L}$&Long-term time-averaged conditional expectation of ${X}_{i}^{\rm L}$  &   $X_{i}^{\rm O}$&Excess value of $Q_i^{\rm O}$    
  \\\hline
  $\bar{X}_{i}^{\rm O}$&Long-term time-averaged conditional expectation of ${X}_{i}^{\rm O}$ & $X_{ji}$&Excess value of $Z_{ji}$ &    $\bar{X}_{ji}$&Long-term time-averaged conditional expectation of ${X}_{ji}$    
  \\\hline
$Y_{i}^{\rm L}$&Square of $X_{i}^{\rm L}$ & $\bar{Y}_{i}^{\rm L}$&Long-term time-averaged conditional expectation of ${Y}_{i}^{\rm L}$  &   $Y_{i}^{\rm O}$&Square of $X_{i}^{\rm O}$    
\\  \hline
  $\bar{Y}_{i}^{\rm O}$&Long-term time-averaged conditional expectation of ${Y}_{i}^{\rm O}$ & $Y_{ji}$&Square of $X_{ji}$  &   $\bar{Y}_{ji}$&Long-term time-averaged conditional expectation of ${Y}_{ji}$    
  \\\hline
   $Z_{ji}$&MEC server $j$'s offloaded-task queue for UE $i$ &  $\beta_i^{\rm L}$ & Related weight of $Q_i^{\rm L}$  &     $\beta_i^{\rm O}$ & Related weight of $Q_i^{\rm O}$ 
   \\\hline
 $\beta_{ij}$&Related weight of $Z_{ji}$ &
  $\tilde{\beta}_i^{\rm O}$&Estimated average of   $\beta_i^{\rm O}$&
  $\tilde{\beta}_{j}$&Estimated
average of   $\beta_{ji}$\\
  \hline
 $\eta_{ij}$ & Association indicator between UE $i$ and MEC server $j$  &   $\boldsymbol{\eta}$ & Network-wide association  vector   &   $\epsilon^{\rm L}_{i}$&Tolerable bound violation probability for $Q_i^{\rm L}$\\
  \hline
$\epsilon^{\rm O}_{i}$&Tolerable bound violation probability for $Q_i^{\rm O}$   &    $\epsilon_{ji}$&Tolerable bound violation probability for $Z_{ji}$    &     $\kappa $&Computation power parameter\\
  \hline
$\lambda_{i}$&UE $i$'s average task arrival rate  & 
$\Psi_i(\eta)$&UE $i$'s utility under a matching $\eta$ &$\Psi_i(\eta)$ & MEC server $j$'s utility under a matching  $\eta$ \\
\hline
$\sigma_i^{\rm L}$&GPD scale parameter of $X_{i}^{\rm L}$  &  $\sigma_i^{\rm L, th}$&Threshold for $\sigma_i^{\rm L}$  &$\sigma_i^{\rm O, th}$&Threshold for the GPD scale parameter of $X_{i}^{\rm O}$  \\
  \hline
 $\sigma_{ji}^{\rm th}$&Threshold for the GPD scale parameter of $X_{ji}$ & $\tau$&Time slot length  &  $\xi_i^{\rm L}$&GPD shape parameter of $X_{i}^{\rm L}$\\
  \hline
 $\xi_i^{\rm L, th}$& Threshold for $\xi_i^{\rm L}$ & $\xi_i^{\rm O, th}$& Threshold for the GPD shape parameter of $X_{i}^{\rm O}$  & $\xi_{ji}^{\rm th}$& Threshold on the GPD shape parameter of $X_{ji}$\\
  \hline
 \end{tabular}
\end {table*}

While conventional communication networks were engineered to boost network capacity, little attention has been paid to reliability and latency performance. Indeed, ultra-reliable and low latency communication (URLLC) is one of the pillars for enabling 5G and is currently receiving significant attention in both academia and industry \cite{URLLC,MehdiURLLC,Qualcomm_URLLC}. 
Regarding  the existing MEC literature,  the vast majority considers the average delay as a performance metric or the quality-of-service requirement \cite{Xu17,Sun17,KB_TWC18,Avg_delay_ICC15, Gilsoo_ICC17,FanLett18,FanJIoT18,FanCog18,Delay_estimation, Zhang_IoTJ17}. In other words, these system designs focus on latency through the lens of the average.
In the works addressing the stochastic nature of the task arrival process \cite{JSAC_queue,Kim18,TaskSplitting,FogPower,Mao_TWC17,EE_stable_GC17}, their prime concern is how to maintain the mean rate stability of task queues,  i.e., ensuring a finite average queue length as time evolves \cite{Neely/Stochastic}.
However, merely focusing on the average-based performance is not sufficient to guarantee URLLC for mission-critical applications, which mandates a further examination in terms of  bound violation probability, high-order statistics, characterization of the extreme events with very low occurrence probabilities, and so forth \cite{MehdiURLLC}. 

The main contribution of this work is to propose a URLLC-centric task offloading and resource allocation framework, by taking into account the statistics of extreme queue length events.
We consider a multi-user MEC architecture with multiple servers having heterogeneous computation resources.
Due to the UE's  limited computation capability and the additional incurred latency during task offloading, the UEs need to 
smartly allocate resources  for local computation and the amount of tasks to offload via wireless transmission
 if the executed applications are latency-sensitive or mission-critical.
Since the queue value is implicitly related to delay, we treat the former as a delay metric in this work.
Motivated by the aforementioned drawbacks of average-based designs, we set a threshold for the queue length and impose a probabilistic requirement on the threshold deviation as a URLLC constraint. In order to model the event of threshold deviation, we characterize its statistics by invoking  {\it extreme value theory} \cite{EVT} and impose another URLLC constraint in terms of higher-order statistics. 
The  problem  is cast as a network-wide power minimization problem for task computation and offloading, subject to statistical URLLC constraints on the  threshold deviation and extreme queue length events.
Furthermore, we incorporate the UEs'  mobility feature and propose a two-timescale  UE-server association and task computation framework. In this regard, taking into account task queue state information,  servers' computation capabilities and workloads, co-channel interference, and URLLC constraints, we associate the UEs  with the MEC servers,  in a long timescale, by utilizing matching theory \cite{matching_theo}.
Then, given the associated MEC server, task offloading and resource allocation are performed in the short timescale. To this end, we leverage Lyapunov stochastic optimization \cite{Neely/Stochastic} to deal with the randomness of task arrivals,  wireless channels, and task queue values. 
Simulation results show that considering the statistics of the extreme queue length as a reliability measure,
the studied partially-offloading scheme includes more reliable task execution than the scheme without MEC servers and the fully-offloading scheme.
In contrast with the received signal strength (RSS)-based baseline, our proposed UE-server association approach achieves better delay performance for heterogeneous MEC server architectures. The performance enhancement is more remarkable in denser networks.

The remainder of this paper is organized as follows. The system model is first specified in Section \ref{Sec: system}.  Subsequently, we formulate the latency requirements, reliability constraints, and the studied optimization problem in Section \ref{Sec: problem}. In Section \ref{Sec: approach}, we detailedly specify the proposed UE-server association mechanism as well as the latency and reliability-aware task offloading and resource allocation framework. The network performance is  evaluated  numerically and discussed in Section \ref{Sec: results} which is followed by Section \ref{Sec: conclude} for conclusions. Furthermore, for the sake of readability, we list all notations in Table \ref{Tab: notations}. The meaning of the notations will be detailedly defined in the following sections.

\section{System Model}\label{Sec: system}

The considered MEC network consists of a set $\mathcal{U}$ of $U$ UEs and  a set $\mathcal{S}$ of $S$ MEC  servers. UEs  have computation capabilities to execute their own tasks locally.
However, due to the limited computation capabilities to execute computation-intense applications,  UEs can wirelessly offload their tasks to the MEC servers with an additional cost of communication latency.
 The MEC servers are equipped with  multi-core CPUs such that different UEs' offloaded tasks can be computed in parallel.
Additionally, the computation and communication timeline is slotted and indexed by $t\in\mathbb{N}$   in which each time slot, with the slot length  $\tau$, is consistent with the coherence block of the wireless channel. 
We further assume that  UEs are randomly distributed and moves continuously in the network, whereas the MEC servers are located in  fixed positions. Since the UE's geographic location keeps changing, the UE is incentivized to offload its tasks to a different server which is closer to the UE, provides a stronger computation capability, or has the lower workload than the currently associated one. In this regard, we consider a two-timescale UE-server association and task-offloading mechanism. Specifically, we group every successive  $T_0$ time slots as a time frame, which is indexed by $n\in\mathbb{Z}^{+}$ and denoted by $\mathcal{T}(n)=[(n-1)T_0,\cdots,nT_0-1]$. In the beginning of each time frame (i.e., the long/slow timescale), each UE is associated with an MEC server. 
Let $\eta_{ij}(n)\in\{0,1\}$ represent the UE-server association indicator in the $n$th time frame, in which $\eta_{ij}(n)=1$ indicates that UE $i$ can offload its tasks to server $j$ during time frame $n$. Otherwise, $\eta_{ij}(n)=0$.  We also assume that each UE can only offload its tasks to one MEC server at a time.  The UE-server association rule can be formulated as
\begin{equation}\label{Eq: UE-server association constraint}
\begin{cases}
\eta_{ij}(n)\in\{0,1\},&\forall\,i\in\mathcal{U},j\in\mathcal{S},
\\\sum\limits_{j\in\mathcal{S}}\eta_{ij}(n)=1,&\forall\, i\in\mathcal{U}.
\end{cases}
\end{equation}
Subsequently in each time slot, i.e., the short/fast timescale, within the $n$th frame, each UE dynamically offloads part of the tasks to the associated MEC server and computes the remaining tasks locally.
The network architecture and timeline of the considered MEC network are shown in Fig.~\ref{Fig: System model}.
\begin{figure}[t]
\centering
\subfigure[System architecture.]{\includegraphics[width=\columnwidth]{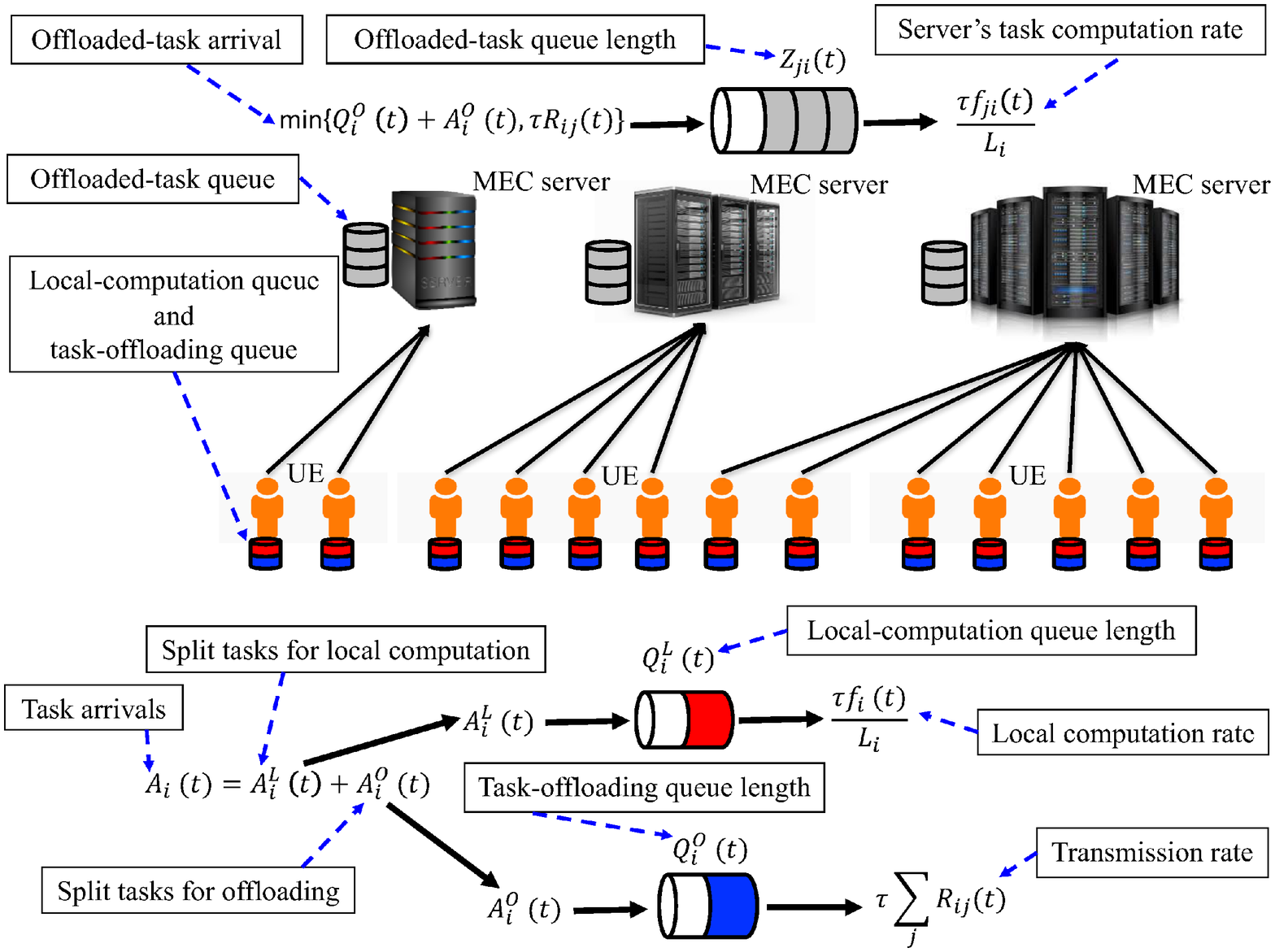}}
\subfigure[Timeline of the two-timescale mechanism.]{\includegraphics[width=\columnwidth]{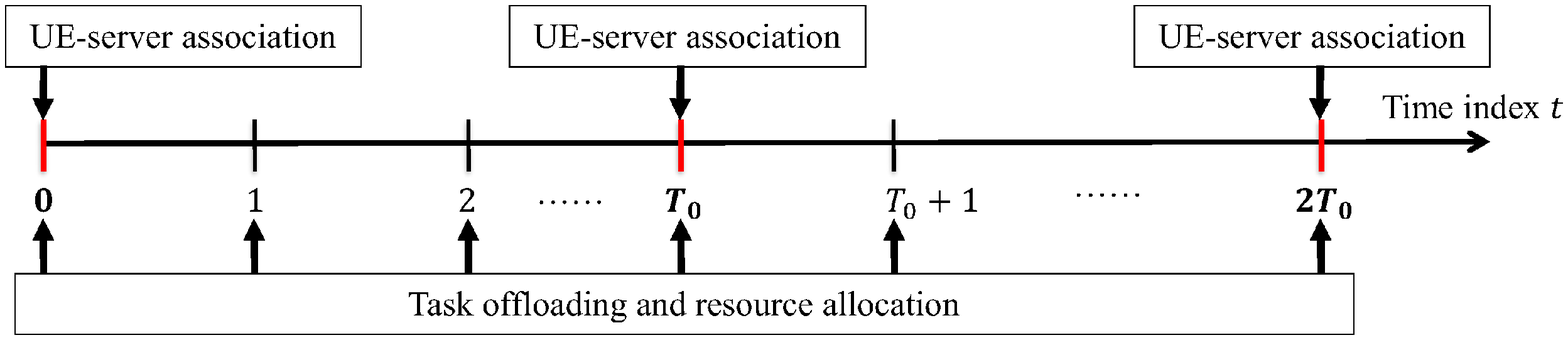}}
	\caption{System model and timeline of the considered MEC network.}
		\label{Fig: System model}
\end{figure}

\subsection{Traffic Model at the UE Side}
The UE uses one application in which tasks arrive in a stochastic manner.
Following the data-partition model \cite{fog_survey}, we assume that  each task can be computed locally, i.e., at the UE, or remotely, i.e., at the server. Different tasks are independent and  can be computed in parallel. Thus, having the task arrivals  $A_i(t)$ in time slot $t$, each UE $i$ divides its arrival into two disjoint parts in which one part $A^{\rm L}_{i}(t)$ is executed locally when the remaining tasks $A^{\rm O}_{i}(t)$ will be offloaded to the server.  Task splitting at UE  $i\in\mathcal{U}$ can be expressed as
\begin{align}\label{Eq: Task splitting constraint}
\begin{cases}
A_i(t)=A^{\rm L}_{i}(t)+A^{\rm O}_{i}(t),
\\ A^{\rm L}_{i}(t),A^{\rm O}_{i}(t)\in\{0,A_{\rm unit},2A_{\rm unit},\cdots\}.
\end{cases}
\end{align}
Here,  $A_{\rm unit}$ represents the unit task which cannot be further split. Moreover, 
we assume that task arrivals are independent and identically distributed (\emph{i.i.d.}) over time  with the average arrival rate  $\lambda_i=\mathbb{E}[A_i]/\tau$.

Each UE has two queue buffers to store the split tasks for local computation and offloading. 
For the local-computation queue of  UE $i\in\mathcal{U}$, the queue length  (in the unit of bits) in time slot $t$ is denoted by $Q_{i}^{\rm L}(t)$ which evolves as
\begin{equation}\label{Eq: local task queue}
Q^{\rm L}_{i}(t+1) =\max\Big\{Q^{\rm L}_{i}(t)+  A^{\rm L}_{i}(t)-\frac{\tau f_{i}(t)}{L_i},0\Big\}.
\end{equation}
Here, $f_i(t)$ (in the unit of cycle/sec) is the UE $i$'s allocated CPU-cycle frequency to execute tasks when $L_i$ accounts for the required CPU cycles per bit for computation, i.e.,  the processing density. The magnitude of the processing density depends on the performed application.\footnote{For example, the six-queen puzzle, 400-frame video game, seven-queen puzzle, face recognition, and virus scanning require the processing densities of 1760\,cycle/bit, 2640\,cycle/bit, 8250\,cycle/bit, 31680\,cycle/bit, and 36992\,cycle/bit, respectively \cite{JSAC_queue}.} Furthermore,
given a CPU-cycle frequency $f_i(t)$, the UE consumes the amount $\kappa [f_i(t)]^3$ of power for computation. $\kappa$ is a parameter affected by the device's hardware implementation \cite{FogPower,CPUhardware}.
For the task-offloading queue of UE $i\in\mathcal{U}$, we denote the queue length  (in the unit of bits) in time slot $t$ as $Q^{\rm O}_{i}(t)$. Analogously, the task-offloading queue dynamics is given by
\begin{equation}\label{Eq: offloading task queue}
Q^{\rm O}_{i}(t+1) =\max\Big\{Q^{\rm O}_{i}(t)+  A^{\rm O}_{i}(t)- \sum\limits_{j\in\mathcal{S}}\tau R_{ij}(t),0\Big\},
\end{equation}
in which
\begin{equation}\label{Eq: rate}
 R_{i j}(t)=W\log_2\bigg(1+\frac{\eta_{ij}(n)P_{i}(t)h_{ij}(t)}{N_0W+\sum\limits_{i'\in\mathcal{U}\setminus i}\eta_{i'j}(n)P_{i'}(t)h_{i'j}(t)}\bigg),
\end{equation}
$\forall\,i\in\mathcal{U},j\in\mathcal{S}$, is UE $i$'s transmission rate\footnote{All transmissions are encoded based on a Gaussian distribution.} to offload tasks to the associated MEC server $j$ in time slot $t\in\mathcal{T}(n)$.
 $P_{i}(t)$ and $N_0$ are UE $i$'s transmit power and the power spectral density of the additive white Gaussian noise (AWGN), respectively. $W$ is the bandwidth dedicated to each server and shared by its associated UEs.
 Additionally, $h_{ij}$ is  the wireless channel gain between UE $i\in\mathcal{U}$ and server $j\in\mathcal{S}$, including path loss and channel fading. We also assume that all channels experience block fading. 
In this work, we mainly consider the uplink, i.e., offloading tasks from the UE to the MEC server, and neglect the downlink, i.e., downloading the computed tasks from the server. 
The rationale is that compared with the offloaded tasks before computation, the computation results typically have smaller sizes \cite{KB_TWC18,Mach_survey,You_TWC17}. Hence, the overheads in the downlink can be neglected.

In order to minimize the total power consumption of resource allocation for local computation and task offloading, the UE adopts the dynamic voltage and frequency scaling (DVFS) capability to adaptively adjust its CPU-cycle frequency \cite{fog_survey,CPUhardware}.
Thus, to allocate the CPU-cycle frequency and transmit power, we impose the following constraints at each UE $i\in\mathcal{U}$, i.e.,
\begin{equation}\label{Eq: UE resource constraint}
\begin{cases}
\kappa[ f_i(t)]^3+P_{i}(t)\leq P_{i}^{\max},
\\f_{i}(t)\geq  0,
\\P_{i}(t)\geq 0,
\end{cases}
\end{equation}
where $P_{i}^{\max}$ is UE $i$'s power budget.

\subsection{Traffic Model at the Server Side}

We assume that each MEC server has distinct queue buffers to store different UEs' offloaded tasks, where the queue length (in bits) of the UE $i$'s offloaded tasks at server $j$ in time slot $t$ is denoted by $Z_{ji}(t)$. The offloaded-task queue length evolves as
\begin{align}
Z_{j i}(t+1)&=\max\Big\{Z_{j i}(t)+  \min\big\{Q^{\rm O}_{i}(t)+  A^{\rm O}_{i}(t),\tau R_{ij}(t)\big\}\notag
\\&\hspace{12em}-\frac{\tau f_{ji}(t)}{L_i} ,0\Big\}\label{Eq: computation queue}
\\ &\leq \max\Big\{Z_{j i}(t)+  \tau R_{ij}(t)-\frac{\tau f_{ji}(t)}{L_i} ,0\Big\},\label{Eq: computation queue ineq}
\end{align}
$\forall\,i\in\mathcal{U},j\in\mathcal{S}$. Here, $f_{ji}(t)$ is the server $j$'s allocated CPU-cycle frequency to process UE $i$'s offloaded tasks.
Note that the MEC server is deployed to provide a faster computation capability for the UE. Thus, we consider the scenario in which each CPU core of the MEC server is dedicated to at most one UE (i.e., its offloaded tasks) in each time slot, and a UE's offloaded tasks at each server can only be computed by one CPU core at a time \cite{TaskSplitting,FogPower}.
The considered computational resource scheduling mechanism at the MEC server is mathematically formulated as 
\begin{equation}\label{Eq: Server resource constraint}
\begin{cases}
\sum\limits_{i\in\mathcal{U}}\mathbbm{1}_{\{f_{ji}(t)>0\}}\leq N_j,&\forall\,j\in\mathcal{S},
\\f_{ji}(t)\in\{0, f_{j}^{\max}\},&\forall\,i\in\mathcal{U},j\in\mathcal{S},
\end{cases}
\end{equation}
where $N_j$ denotes the total CPU-core number of server $j$, $f_{j}^{\max}$ is server $j$'s computation capability of one CPU core, and $\mathbbm{1}_{\{\cdot\}}$ is  the indicator function. 
In \eqref{Eq: Server resource constraint}, we account for the allocated CPU-cyle frequencies to all UEs even though some UEs are not associated with this server in the current time frame. The rationale will be detailedly explained in Section \ref{Sec: Server resource scheduling} after  formulating the concerned optimization problem.
Additionally,  in order to illustrate the relationship between the offloaded-task queue length and the transmission rate, we introduce inequality \eqref{Eq: computation queue ineq} which will be further used to formulate the latency and reliability requirements of the considered MEC system and derive the solution of the studied optimization problem.

\section{Latency Requirements, Reliability Constraints, and Problem Formulation}\label{Sec: problem}

In this work, the end-to-end delays experienced by the locally-computed tasks $A^{\rm L}_i(t)$ and offloaded tasks $A^{\rm O}_i(t)$ consist of different components. When the task is computed locally, it experiences  the queuing delay (for computation) and computation delay at the UE. If the task is offloaded to the MEC server, the end-to-end delay includes: 1) queuing delay (for offloading) at  the UE, 2) wireless transmission delay while offloading, 3) queuing delay (for computation) at the server, and 4) computation delay at the server.
From Little's law, we know that the average queuing delay is proportional to the average queue length  \cite{Littlelaw}.
However, without taking the tail distribution of the queue length into account, solely focusing on the average queue length fails to account for the low-latency and reliability requirement \cite{MehdiURLLC}. To tackle this, we focus on the statistics of the task queue and impose  probabilistic constraints on the local-computation and task-offloading queue lengths of each UE $i\in\mathcal{U}$ as follows:
\begin{align}
 &\lim\limits_{T\to\infty}\frac{1}{T}\sum\limits_{t=1}^{T}\Pr\big(Q^{\rm L}_{i}(t)> d^{\rm L}_{i} \big)\leq \epsilon^{\rm L}_{i},\label{Eq: Violation-Loc-Prob}
 \\& \lim\limits_{T\to\infty}\frac{1}{T}\sum\limits_{t=1}^{T}\Pr\big(Q^{\rm O}_{i}(t)> d^{\rm O}_{i} \big)\leq \epsilon^{\rm O}_{i}.\label{Eq: Violation-Off-Prob}
\end{align}
Here, $d^{\rm L}_{i}$ and $d^{\rm O}_{i}$ are the  queue length bounds when $ \epsilon^{\rm L}_{i}\ll 1$ and $ \epsilon^{\rm O}_{i}\ll1$ are the tolerable  bound violation probabilities.
Furthermore, the queue length bound violation also undermines the reliability issue of task computation. 
For example, if a finite-size queue buffer is over-loaded, the incoming tasks will be dropped.

In addition to the bound violation probability, let us look at the complementary cumulative distribution function (CCDF) of the UE's local-computation queue length, i.e., $\bar{F}_{Q^{\rm L}_i}(q)=\Pr\big(Q^{\rm L}_{i}> q \big)$, which reflects the queue length profile. If the monotonically decreasing  CCDF decays faster while increasing $q$, the probability of having an extreme queue length is lower.
Since the prime concern in this work lies in the extreme-case events with very low occurrence  probabilities, i.e., $\Pr\big(Q^{\rm L}_{i}(t)> d^{\rm L}_{i} \big)\ll 1$, we resort to principles of extreme value theory\footnote{Extreme value theory is a powerful and robust framework to study the tail behavior of a distribution. Extreme value theory also  provides statistical models for the computation of extreme risk measures.} to characterize the statistics and tail distribution of the extreme event $Q^{\rm L}_{i}(t)> d^{\rm L}_{i}$.  
To this end, we first introduce the \emph{Pickands–Balkema–de Haan theorem} \cite{EVT}. 
\begin{theorem}[{\bf Pickands–Balkema–de Haan theorem}]\label{Thm: Pareto}
Consider a random variable $Q$, with the cumulative distribution function (CDF) $F_{Q}(q)$, and a threshold value $d$. 
As the threshold $d$ closely approaches $F^{-1}_{Q}(1)$, i.e., $d\to\sup\{q\!\!: F_{Q}(q)<1\}$, the conditional CCDF of the excess value $X|_{Q>d}=Q-d>0$, i.e.,  $\bar{F}_{X|Q>d}(x)=\Pr(Q-d> x|Q>d)$, can be approximated by a generalized Pareto distribution (GPD) $G(x;\sigma,\xi)$, i.e.,
 \begin{align}
 &\bar{F}_{X|Q>d}(x)\approx G(x;\sigma,\xi)\notag
 \\&\coloneqq
\begin{cases}
  \Big(1+\frac{\xi x}{\sigma}\Big)^{-1/\xi},
 &\mbox{where } x\geq 0\mbox{ and } \xi>0,
\\  e^{-x/\sigma}, &\mbox{where } x\geq 0 \mbox{ and } \xi=0,
\\  \Big(1+\frac{\xi x}{\sigma}\Big)^{-1/\xi},
 &\mbox{where } 0\leq  x\leq -\sigma/\xi \mbox{ and } \xi<0,
\end{cases} \label{Eq: Theo GPD}
\end{align}
which is characterized by a scale parameter $\sigma >0$ and a shape parameter $\xi \in\mathbb{R}$. 
 \end{theorem}
In other words, the conditional CCDF of the excess value $X|_{Q>d}$ converges to a GPD as $d\to\infty$. However, from the proof \cite{EVT} for Theorem \ref{Thm: Pareto}, we know that the GPD provides a good approximation when $F_{Q}(d)$ is close to 1, e.g., $F_{Q}(d)=0.99$. That is, depending on the CDF of $Q$, imposing a very large $d$ might not be necessary for obtaining the approximated GPD.
Moreover, for a GPD $G(x;\sigma,\xi)$, its mean  $\sigma/(1-\xi)$ and other higher-order statistics such as variance  $\frac{\sigma^2}{(1-\xi)^2(1-2\xi)}$ and skewness exist if $\xi<1$, $\xi<\frac{1}{2}$, and $\xi<\frac{1}{3}$, respectively. 
Note that the scale parameter $\sigma$ and the domain $x$ of $G(x;\sigma,\xi)$ are in the same order. In this regard, we can see that $G(\sigma;\sigma,0)=e^{-1}=0.37$ at $x=\sigma$ and $G(3\sigma;\sigma,0)=e^{-3}=0.05$ at $x=3\sigma$ in  \eqref{Eq: Theo GPD}.
 We also show the CCDFs of the GPDs for various shape parameters $\xi$ in Fig.~\ref{Fig: GPD case}, where the x-axis is indexed with respect to the normalized value $x/\sigma$. As shown in Fig.~\ref{Fig: GPD case}, the decay speed of the CCDF increases as $\xi$ decreases.  In contrast with the curves with $\xi\geq0$, we can see that the CCDF decays rather sharply when $\xi\leq -3$.
\begin{figure}[t]
\centering
\includegraphics[width=\columnwidth]{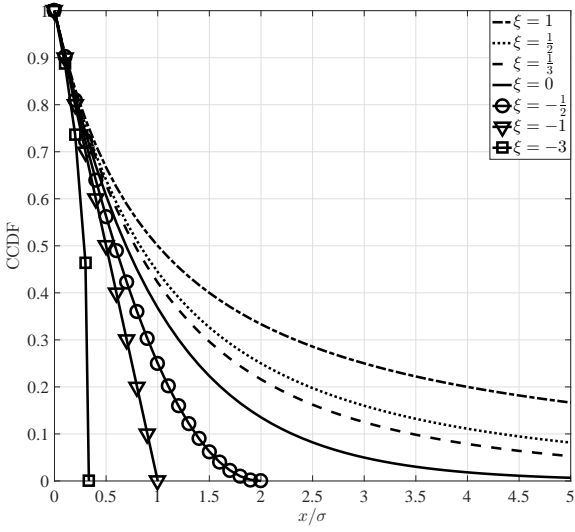}
	\caption{CCDFs of the GPDs for various shape pamameters $\xi$.}
\label{Fig: GPD case}
\end{figure}

Now, let us denote the  excess value (with respect to the threshold $d_i^{\rm L}$ in \eqref{Eq: Violation-Loc-Prob})  of the local-computation queue of each UE $i\in\mathcal{U}$ in time slot $t$ as  $X^{\rm L}_i(t)|_{Q^{\rm L}_{i}(t)> d^{\rm L}_i}=Q^{\rm L}_{i}(t)- d^{\rm L}_{i}>0$. 
By applying Theorem \ref{Thm: Pareto},  the  excess queue value can be approximated by a GPD $G(x_i;\sigma^{\rm L}_{i},\xi^{\rm L}_{i})$ whose mean and variance  are
\begin{align}
\mathbb{E}\big[X^{\rm L}_{i}(t)|Q^{\rm L}_{i}(t)> d^{\rm L}_{i} \big]& \approx\frac{\sigma^{\rm L}_{i}}{1-\xi^{\rm L}_{i}},\label{Eq:remodel_Pareto_mean}
\\\mbox{Var}\big(X^{\rm L}_{i}(t) |Q^{\rm L}_{i}(t)> d^{\rm L}_{i}\big)& \approx\frac{(\sigma^{\rm L}_{i})^2}{(1-\xi^{\rm L}_{i})^2(1-2\xi^{\rm L}_{i})},\label{Eq:remodel_Pareto_variance}
\end{align}
with the  corresponding scale parameter   $\sigma^{\rm L}_{i}$  and shape parameter $\xi^{\rm L}_{i}$.
In \eqref{Eq:remodel_Pareto_mean} and \eqref{Eq:remodel_Pareto_variance}, we can find that the smaller  $\sigma^{\rm L}_{i}$ and $\xi^{\rm L}_{i}$ are, the smaller the mean value and variance. Since the approximated GPD is just characterized by the  scale  and  shape parameters as mentioned previously, therefore,
we impose thresholds on these two parameters, i.e.,  $\sigma^{\rm L}_{i}\leq \sigma_{i}^{\rm L, th}$ and $\xi^{\rm L}_{i}\leq \xi_{i}^{\rm L, th}$. 
The selection of threshold values can be referred to the above discussions about the GPD, Fig.~\ref{Fig: GPD case}, and the magnitude of the interested metric's values.
Subsequently, applying the two parameter thresholds and $\mbox{Var}(X^{\rm L}_{i})=\mathbb{E}[(X^{\rm L}_{i})^2]-\mathbb{E}[X^{\rm L}_{i}]^2$ to \eqref{Eq:remodel_Pareto_mean} and \eqref{Eq:remodel_Pareto_variance}, we consider the constraints on the long-term time-averaged  conditional mean and second moment of the excess value of each UE's local-computation queue length, i.e., $\forall\,i\in\mathcal{U}$,
\begin{align}
\bar{X}^{\rm L}_{i}& =\lim\limits_{T\to\infty}\frac{1}{T}\sum\limits_{t=1}^{T}\mathbb{E}\big[X^{\rm L}_{i} (t)|Q^{\rm L}_{i} (t)> d^{\rm L}_{i} \big]\leq \frac{\sigma_i^{\rm L, th}}{1-\xi_i^{\rm L, th}},\label{Eq: GPD-Loc-mean}
\\\bar{Y}^{\rm L}_{i}&=\lim\limits_{T\to\infty}\frac{1}{T}\sum\limits_{t=1}^{T}\mathbb{E}\big[Y^{\rm L}_{i}(t) |Q^{\rm L}_{i}(t)> d^{\rm L}_{i}\big]\notag
\\&\leq \frac{2\big(\sigma_i^{\rm L, th}\big)^2}{\big(1-\xi_i^{\rm L, th}\big)\big(1-2\xi_i^{\rm L, th}\big)},\label{Eq: GPD-Loc-var}
\end{align}
 with $Y^{\rm L}_{i}(t)= [X^{\rm L}_{i}(t)]^2$.
Analogously, denoting  the  excess value, with respect to the threshold $d_i^{\rm O}$, of UE  $i$'s task-offloading queue length in time slot $t$ as  $X^{\rm O}_i(t)|_{Q^{\rm O}_{i}(t)> d^{\rm O}_i}=Q^{\rm O}_{i}(t)- d^{\rm O}_{i}>0$,
we have the constraints on the long-term time-averaged  conditional mean and second moment
\begin{align}
\bar{X}^{\rm O}_{i} &=\lim\limits_{T\to\infty}\frac{1}{T}\sum\limits_{t=1}^{T}\mathbb{E}\big[X^{\rm O}_{i}(t)|Q^{\rm O}_{i}(t)> d^{\rm O}_{i} \big]
\leq \frac{\sigma_i^{\rm O, th}}{1-\xi_i^{\rm O, th}},\label{Eq: GPD-Off-mean}
\\\bar{Y}^{\rm O}_{i} &=\lim\limits_{T\to\infty}\frac{1}{T}\sum\limits_{t=1}^{T}\mathbb{E}\big[Y^{\rm O}_{i}(t) |Q^{\rm O}_{i}(t)> d^{\rm O}_{i}\big]\notag
\\&\leq \frac{2\big(\sigma_i^{\rm O, th}\big)^2}{\big(1-\xi_i^{\rm O, th}\big)\big(1-2\xi_i^{\rm O, th}\big)},\label{Eq: GPD-Off-var}
\end{align}
for all UEs $i\in\mathcal{U}$,
in which $ \sigma_{i}^{\rm O, th}$ and $\xi_{i}^{\rm O, th}$ are the  thresholds for the characteristic parameters of the approximated GPD, and  $Y^{\rm O}_{i}(t)= [X^{\rm O}_{i}(t)]^2$.

Likewise, the average queuing delay at the server is proportional to the ratio of the average queue length to the average transmission rate. Referring to \eqref{Eq: computation queue ineq}, we consider the probabilistic constraint as follows:
\begin{align}\label{Eq: Violation-Server-Prob}
\lim\limits_{T\to\infty}\frac{1}{T}\sum\limits_{t=1}^{T}\Pr\bigg(\frac{Z_{ji}(t)}{\tilde{R}_{ij}(t-1)}> d_{ji}\bigg)\leq \epsilon_{ji},
\end{align}
$\forall\,i\in\mathcal{U},j\in\mathcal{S}$, with the threshold  $d_{ji}$ and  tolerable violation probability $\epsilon_{ji}\ll 1$, on the offloaded-task queue length at the MEC server.
$\tilde{R}_{ij}(t-1)=\frac{1}{t}\sum_{\tau=0}^{t-1}R_{ij}(\tau)$ is the moving time-averaged transmission rate.
Similar to the task queue lengths at the UE side,
 we further denote the excess value,  with respect to the threshold $\tilde{R}_{ij}(t-1)d_{ji}$, in time slot $t$ as  $X_{ji}(t)|_{Z_{ji}(t)>\tilde{R}_{ij}(t-1) d_{ji}}=Z_{ji}(t)- \tilde{R}_{ij}(t-1)d_{ji}>0$ of the offloaded-task queue length of UE $i\in\mathcal{U}$ at server $j\in\mathcal{S}$ and impose the constraints as follows:
\begin{align}
\bar{X}_{ji} =\lim\limits_{T\to\infty}\frac{1}{T}\sum\limits_{t=1}^{T}\mathbb{E}\big[X_{ji}(t)|Z_{ji}(t)>\tilde{R}_{ij}(t-1) d_{ji} \big]&\notag
\\\leq \frac{\sigma_{ji}^{\rm th}}{1-\xi_{ji}^{\rm th}}&,\label{Eq: GPD-Server-mean}
\\\bar{Y}_{ji} =\lim\limits_{T\to\infty}\frac{1}{T}\sum\limits_{t=1}^{T}\mathbb{E}\big[Y_{ji}(t) |Z_{ji}(t)> \tilde{R}_{ij}(t-1)d_{ji}\big] &\notag
\\\leq \frac{2(\sigma_{ji}^{\rm th})^2}{(1-\xi_{ji}^{\rm th})(1-2\xi_{ji}^{\rm th})}&,\label{Eq: GPD-Server-var}
\end{align}
with $Y_{ji}(t)= [X_{ji}(t)]^2$. Here, $ \sigma_{ji}^{\rm th}$ and $\xi_{ji}^{\rm th}$ are the  thresholds for the characteristic parameters of the approximated GPD.

We note that the local computation delay at the UE and the transmission delay while offloading are inversely proportional to the computation speed $f_i(t)/L_i$ and the transmission rate $\sum_{j\in\mathcal{S}}R_{ij}(t)$ as per \eqref{Eq: local task queue} and \eqref{Eq: offloading task queue}, respectively. To decrease the local computation and transmission delays, the UE should allocate a higher local CPU-cycle frequency and more transmit power, which, on the other hand, incurs  energy shortage. 
Since allocating a higher CPU-cycle frequency and more transmit power can also further decrease the queue length, both (local computation and transmission) delays are implicitly taken into account in the queue length constraints \eqref{Eq: Violation-Loc-Prob}, \eqref{Eq: Violation-Off-Prob}, and \eqref{Eq: GPD-Loc-mean}--\eqref{Eq: GPD-Off-var}. 
At the server side,  the remote computation delay can be neglected because one CPU core with the better computation capability is dedicated to one UE's offloaded tasks at a time. On the other hand, the server needs to schedule its computational resources, i.e., multiple CPU cores, when the  associated UEs are more than the CPU cores.

Incorporating the aforementioned latency requirements and reliability constraints, 
the studied optimization problem is formulated as follows:
\begin{IEEEeqnarray*}{rcl}
\mbox{\bf MP:}~~&\underset{\boldsymbol{\eta}(n), \mathbf{f}(t),\mathbf{P}(t)}{\mbox{minimize}}&~~ \sum\limits_{i\in\mathcal{U}}(\bar{P}^{\rm C}_{i}+\bar{P}^{\rm T}_{i})\notag
\\&\mbox{subject to}&~~\mbox{\eqref{Eq: UE-server association constraint} for UE-server association,}\notag
\\&&~~\mbox{\eqref{Eq: Task splitting constraint} for task splitting,}\notag
\\&&~~\mbox{\eqref{Eq: UE resource constraint} and \eqref{Eq: Server resource constraint} for resource allocation,}\notag
\\&&~~\mbox{\eqref{Eq: Violation-Loc-Prob}, \eqref{Eq: Violation-Off-Prob}, and \eqref{Eq: Violation-Server-Prob} for}\notag
\\&&~~\mbox{\hspace{4em}queue length bound violation,}\notag
\\&&~~\mbox{\eqref{Eq: GPD-Loc-mean}--\eqref{Eq: GPD-Off-var}, \eqref{Eq: GPD-Server-mean}, and  \eqref{Eq: GPD-Server-var} for the GPDs,}
\end{IEEEeqnarray*}
where $\bar{P}^{\rm C}_{i}=\lim\limits_{T\to\infty}\frac{1}{T}\sum_{t=0}^{T-1}\kappa [f_i(t)]^3$ and $\bar{P}^{\rm T}_{i}=\lim\limits_{T\to\infty}\frac{1}{T}\sum_{t=0}^{T-1}P_{i}(t)$ are the UE $i$'s long-term time-averaged  power consumptions for local computation and task offloading, respectively.  $\boldsymbol{\eta}(n)=(\eta_{ij}(n):i\in\mathcal{U},j\in\mathcal{S})$ and $\mathbf{P}(t)=(P_{i}(t)\!\!:i\in\mathcal{U})$ denote the network-wide UE-server association and transmit power allocation vectors, respectively.  In addition, $\mathbf{f}(t)=(f_i(t),\mathbf{f}_{j}(t)\!\!:i\in\mathcal{U},j\in\mathcal{S})$ denotes the network-wide computational resource allocation vector in which  $\mathbf{f}_j(t)=(f_{ji}(t)\!\!:i\in\mathcal{U},j\in\mathcal{S})$  is the computational resource allocation vector of server $j$. To solve problem {\bf MP}, we utilize  techniques from Lyapunov stochastic optimization and  propose a dynamic task offloading and resource allocation policy in the next section.

\section{Latency and Reliability-Aware Task Offloading and Resource Allocation}\label{Sec: approach}

Let us give an overview of the proposed task offloading and resource allocation  approach before specifying the details.
In the beginning of each time frame, i.e., every $T_0$ slots, we carry out a UE-server association, taking into account the wireless link strength, the UEs' and servers' computation capabilities, their historical workloads, and  URLLC constraints \eqref{Eq: Violation-Off-Prob} and 
 \eqref{Eq: GPD-Off-mean}--\eqref{Eq: GPD-Server-var}. To this end, a many-to-one matching algorithm is utilized to associate each server with multiple UEs. Afterwards, we focus on task offloading and resource allocation by solving three decomposed optimization problems, via Lyapunov optimization, in each time slot. At the UE side, each UE splits its instantaneous task arrivals into two parts, which will be computed locally and offloaded respectively, while allocating the local computation CPU-cyle frequency and transmit power for offloading. At the server side, each MEC server schedules its CPU cores to execute the UEs' offloaded tasks. 
 In the procedures (of task splitting and offloading, resource allocation, and CPU-core scheduling), the related URLLC constraints out of \eqref{Eq: Violation-Loc-Prob},  \eqref{Eq: Violation-Off-Prob}, and \eqref{Eq: GPD-Loc-mean}--\eqref{Eq: GPD-Server-var}  are considered. The details of our proposed approach  will be illustrated in the remainder of this section.

\subsection{Lyapunov Optimization Framework}

We first introduce a virtual queue $Q^{\rm L, (X)}_{i}$ for the long-term time-averaged constraint \eqref{Eq: GPD-Loc-mean} with the queue evolution as follows:
\begin{align}
&Q^{\rm L, (X)}_{i}(t+1) =\max\bigg\{Q^{\rm L, (X)}_{i}(t)\notag
\\&\quad+\bigg(X^{\rm L}_{i}(t+1)-\frac{\sigma_{i}^{\rm L, th}}{1-\xi_{i}^{\rm L, th}}\bigg)\times\mathbbm{1}_{\{Q^{\rm L}_{i}(t+1)> d^{\rm L}_{i}\}},0\bigg\},\label{Eq: virtual GPD_Loc-1}
\end{align}
in which the incoming traffic amount $X^{\rm L}_{i}(t+1)\times\mathbbm{1}_{\{Q^{\rm L}_{i}(t+1)> d^{\rm L}_{i}\}}$ and outgoing traffic amount $\frac{\sigma_{i}^{\rm L, th}}{1-\xi_{i}^{\rm L, th}}\times\mathbbm{1}_{\{Q^{\rm L}_{i}(t+1)> d^{\rm L}_{i}\}}$ correspond to the left-hand side and right-hand side of the inequality  \eqref{Eq: GPD-Loc-mean}, respectively.
Note that \cite{Neely/Stochastic} ascertains that the introduced virtual queue  is \emph{mean rate stable}, i.e., $\lim\limits_{t\to\infty}\frac{\mathbb{E}[|Q^{\rm L, (X)}_{i}(t)|]}{t}\to 0$, is equivalent to satisfying the long-term time-averaged constraint  \eqref{Eq: GPD-Loc-mean}.
Analogously, 
for the constraints  \eqref{Eq: GPD-Loc-var}--\eqref{Eq: GPD-Off-var}, \eqref{Eq: GPD-Server-mean}, and \eqref{Eq: GPD-Server-var}, we respectively introduce the virtual queues as follows:
\begin{align}
&Q^{\rm L, (Y)}_{i}(t+1) =\max\bigg\{Q^{\rm L, (Y)}_{i}(t)+\bigg(Y^{\rm L}_{i}(t+1)\notag
\\&- \frac{2\big(\sigma_{i}^{\rm L, th}\big)^2}{\big(1-\xi_{i}^{\rm L, th}\big)\big(1-2\xi_{i}^{\rm L, th}\big)}\bigg)
\times\mathbbm{1}_{\{Q^{\rm L}_{i}(t+1)> d^{\rm L}_{i}\}},0\bigg\} ,\label{Eq: virtual GPD_Loc-2}
\end{align}
\begin{align}
&Q^{\rm O, (X)}_{i}(t+1) =\max\bigg\{Q^{\rm O, (X)}_{i}(t)\notag
\\&+\bigg(X^{\rm O}_{i}(t+1)-\frac{\sigma_{i}^{\rm O, th}}{1-\xi_{i}^{\rm O, th}}\bigg)\times\mathbbm{1}_{\{Q^{\rm O}_{i}(t+1)> d^{\rm O}_{i}\}},0\bigg\},\label{Eq: virtual GPD_Off-1}
\\&Q^{\rm O, (Y)}_{i}(t+1) =\max\bigg\{Q^{\rm O, (Y)}_{i}(t)+\bigg(Y^{\rm O}_{i}(t+1)\notag
\\&- \frac{2\big(\sigma_{i}^{\rm O, th}\big)^2}{\big(1-\xi_{i}^{\rm O, th}\big)\big(1-2\xi_{i}^{\rm O, th}\big)}\bigg)\times\mathbbm{1}_{\{Q^{\rm O}_{i}(t+1)> d^{\rm O}_{i}\}},0\bigg\} ,\label{Eq: virtual GPD_Off-2}
\\&Q^{(\rm X)}_{ji}(t+1) =\max\bigg\{Q^{(\rm X)}_{ji}(t)+\bigg(X_{ji}(t+1)\notag
\\&- \frac{\sigma_{ji}^{\rm th}}{1-\xi_{ji}^{\rm th}}\bigg)\times\mathbbm{1}_{\{Z_{ji}(t+1)> \tilde{R}_{ij}(t)d_{ji}\}},0\bigg\},\label{Eq: virtual GPD_server-1}
\\&Q^{(\rm Y)}_{ji}(t+1) =\max\bigg\{Q^{(\rm Y)}_{ji}(t)+\bigg(Y^k_{ji}(t+1)\notag
\\&- \frac{2(\sigma_{ji}^{\rm th})^2}{(1-\xi_{ji}^{\rm th})(1-2\xi_{ji}^{\rm th})}\bigg)\times\mathbbm{1}_{\{Z_{ji}(t+1)> \tilde{R}_{ij}(t)d_{ji}\}},0\bigg\}.\label{Eq: virtual GPD_server-2}
\end{align}
Additionally, given an event  $B$ and the set of all possible outcomes $ \Omega$, we can derive
$\mathbb{E}[\mathbbm{1}_{\{B\}}]=1\cdot \Pr(B)+0\cdot \Pr(\Omega\setminus B)=\Pr(B)$.
By applying $\mathbb{E}[\mathbbm{1}_{\{B\}}]=\Pr(B)$, constraints \eqref{Eq: Violation-Loc-Prob}, \eqref{Eq: Violation-Off-Prob}, and \eqref{Eq: Violation-Server-Prob} can be equivalently rewritten as
\begin{align}
 \lim\limits_{T\to\infty}\frac{1}{T}\sum\limits_{t=1}^{T}\mathbb{E}\big[\mathbbm{1}_{\{Q^{\rm L}_{i}(t)> d^{\rm L}_{i}\}}\big]&\leq \epsilon^{\rm L}_{i},\label{Eq: remodel_Violation-Loc-Prob}
\\ \lim\limits_{T\to\infty}\frac{1}{T}\sum\limits_{t=1}^{T}\mathbb{E}\big[\mathbbm{1}_{\{Q^{\rm O}_{i}(t)> d^{\rm O}_{i}\}}\big]&\leq \epsilon^{\rm O}_{i},\label{Eq: remodel_Violation-Off-Prob}
\\ \lim\limits_{T\to\infty}\frac{1}{T}\sum\limits_{t=1}^{T}\mathbb{E}\big[\mathbbm{1}_{\{Z_{ji}(t)> \tilde{R}_{ij}(t-1)d_{ji}\}}\big]&\leq \epsilon_{ji}.\label{Eq: remodel_Violation-Server-Prob}
\end{align}
Then let us follow the above steps. The corresponding virtual queues of \eqref{Eq: remodel_Violation-Loc-Prob}--\eqref{Eq: remodel_Violation-Server-Prob} are expressed as
%
%
\begin{align}
&Q^{\rm L, (Q)}_{i}(t+1) =\max\Big\{Q^{\rm L, (Q)}_{i}(t)
+\mathbbm{1}_{\{Q^{\rm L}_{i}(t+1)> d^{\rm L}_{i}\}} -  \epsilon^{\rm L}_{i},0\Big\},\label{Eq: virtual violation_Loc-1}
\\&Q^{\rm O, (Q)}_{i}(t+1) =\max\Big\{Q^{\rm O, (Q)}_{i}(t)
+\mathbbm{1}_{\{Q^{\rm O}_{i}(t+1)> d^{\rm O}_{i}\}} -  \epsilon^{\rm O}_{i},0\Big\},\label{Eq: virtual violation_Off-1}
\\&Q^{(\rm Z)}_{ji}(t+1) =\max\Big\{Q^{(\rm Z)}_{ji}(t)
+\mathbbm{1}_{\{Z_{ji}(t+1)>\tilde{R}_{ij}(t) d_{ji}\}}-  \epsilon_{ji},0\Big\}.\label{Eq: virtual violation_server-1}
\end{align}
%
%
%
%
Now problem {\bf MP} is equivalently transferred to  \cite{Neely/Stochastic}
 \begin{IEEEeqnarray*}{rcl}
\mbox{\bf MP':}~~&\underset{\boldsymbol{\eta}(n), \mathbf{f}(t),\mathbf{P}(t)}{\mbox{minimize}}&~~ \sum\limits_{i\in\mathcal{U}}(\bar{P}^{\rm C}_{i}+\bar{P}^{\rm T}_{i})\notag
\\&\mbox{subject to}&~~\mbox{\eqref{Eq: UE-server association constraint}, \eqref{Eq: Task splitting constraint}, \eqref{Eq: UE resource constraint}, and \eqref{Eq: Server resource constraint}},\notag
\\&&~~\mbox{Stability of \eqref{Eq: virtual GPD_Loc-1}--\eqref{Eq: virtual GPD_server-2} and \eqref{Eq: virtual violation_Loc-1}--\eqref{Eq: virtual violation_server-1}}.
\end{IEEEeqnarray*}
To solve problem {\bf MP'}, we let $\mathbf{Q}(t)=\big(Q^{\rm L, (X)}_{i}(t),Q^{\rm L, (Y)}_{i}(t),Q^{\rm O, (X)}_{i}(t),Q^{\rm O, (Y)}_{i}(t),Q^{(\rm X)}_{ji}(t),Q^{(\rm Y)}_{ji}(t),$
\linebreak
$Q^{\rm L, (Q)}_{i}(t),Q^{\rm O, (Q)}_{i}(t),Q^{(\rm Z)}_{ji}(t)\!\!:i\in\mathcal{U},j\in\mathcal{S}\big)$  denote the combined queue vector for notational simplicity and express the conditional Lyapunov drift-plus-penalty for slot $t$ as
\begin{equation}\label{Eq: Conditional Lyapunov drift}
 \mathbb{E}\Big[\mathcal{L}(\mathbf{Q}(t+1))-\mathcal{L}(\mathbf{Q}(t))+\sum\limits_{i\in\mathcal{U}}V\big(\kappa [f_i(t)]^3+P_{i}(t)\big)
\Big|\mathbf{Q}(t)\Big],
\end{equation}
where 
\begin{align*}
&\mathcal{L}(\mathbf{Q}(t))=\frac{1}{2}\sum_{i\in\mathcal{U}}\Big(\big[Q^{\rm L, (X)}_{i}(t)\big]^2+\big[Q^{\rm L, (Y)}_{i}(t)\big]^2+\big[Q^{\rm O, (X)}_{i}(t)\big]^2
\\&\quad+\big[Q^{\rm O, (Y)}_{i}(t)\big]^2+\big[Q^{\rm L, (Q)}_{i}(t)\big]^2+\big[Q^{\rm O, (Q)}_{i}(t)\big]^2\Big)
\\&\quad+\frac{1}{2}\sum_{i\in\mathcal{U}}\sum_{j\in\mathcal{S}}\Big(\big[Q^{(\rm X)}_{ji}(t)\big]^2+\big[Q^{(\rm Y)}_{ji}(t)\big]^2+\big[Q^{(\rm Z)}_{ji}(t)\big]^2\Big)
\end{align*}
is the Lyapunov function. The term $V\geq 0$ is a parameter which trades off objective optimality and queue length reduction. 
Subsequently, plugging the inequality $(\max\{x,0\})^2\leq x^2$, all physical and virtual queue dynamics, and  \eqref{Eq: computation queue ineq} into  \eqref{Eq: Conditional Lyapunov drift},  we can derive
\begin{align}
& \eqref{Eq: Conditional Lyapunov drift}\leq C+ \mathbb{E}\bigg[\sum_{i\in\mathcal{U}}\Big[\Big(Q^{\rm L, (X)}_{i}(t)+Q^{\rm L}_{i}(t)+2Q^{\rm L, (Y)}_{i}(t)Q^{\rm L}_{i}(t)\notag
 \\&+2\big[Q^{\rm L}_{i}(t)\big]^3 \Big)\Big(  A^{\rm L}_{i}(t)-\frac{\tau f_{i}(t)}{L_i}\Big)\times\mathbbm{1}_{\{Q^{\rm L}_{i}(t)+  A_{i}(t)> d^{\rm L}_{i}\}}\notag
\\&+Q^{\rm L, (Q)}_{i}(t)\times\mathbbm{1}_{\{\max\{Q^{\rm L}_{i}(t)+  A^{\rm L}_{i}(t)-\tau f_{i}(t)/L_i,0\}> d^{\rm L}_{i}\}}\Big]\notag
\\&+\sum_{i\in\mathcal{U}}\Big[\Big(Q^{\rm O, (X)}_{i}(t)+Q^{\rm O}_{i}(t)+2Q^{\rm O, (Y)}_{i}(t)Q^{\rm O}_{i}(t)\notag
\\&+2\big[Q^{\rm O}_{i}(t)\big]^3\Big)\Big( A^{\rm O}_{i}(t)- \sum\limits_{j\in\mathcal{S}}\tau R_{ij}(t)\Big)\times\mathbbm{1}_{\{Q^{\rm O}_{i}(t)+  A_{i}(t)> d^{\rm O}_{i}\}}\notag
\\&+Q^{\rm O, (Q)}_{i}(t)\times\mathbbm{1}_{\{\max\{Q^{\rm O}_{i}(t)+  A^{\rm O}_{i}(t)- \sum\limits_{j\in\mathcal{S}}\tau R_{ij}(t),0\}> d^{\rm O}_{i}\}}\Big]\notag
\\&+\sum_{i\in\mathcal{U}}\sum_{j\in\mathcal{S}}\Big[\Big(Q^{(\rm X)}_{ji}(t)+Z_{ji}(t)+2Q^{(\rm Y)}_{ji}(t)Z_{ji}(t)+2\big[Z_{ji}(t)\big]^3\Big)\notag
\\&\Big(  \tau R_{ij}(t)-\frac{\tau f_{ji}(t)}{L_i} \Big)\times\mathbbm{1}_{\{Z_{ji}(t)+  \tau R^{\max}_{ij}(t)>\tilde{R}_{ij}(t-1) d_{ji}\}}\notag
\\&+Q^{(\rm Z)}_{ji}(t)\times\mathbbm{1}_{\{ \max\{Z_{ji}(t)+  \tau R_{ij}(t)-\tau f_{ji}(t)/L_i ,0\}>\tilde{R}_{ij}(t-1) d_{ji}\}}\Big]\notag
\\&+\sum\limits_{i\in\mathcal{U}}V\big(\kappa [f_i(t)]^3+P_{i}(t)\big)\Big|\mathbf{Q}(t)\bigg].\label{Eq: Lyapunov bound} 
\end{align}
Here,  $R_{ij}^{\max}(t)= W\log_2\big(1+\frac{P_{i}^{\max}h_{ij}(t)}{N_0W}\big)$ is UE $i$'s maximum offloading rate. Since the constant $C$ does not affect the system performance in Lyapunov optimization, we omit its details in \eqref{Eq: Lyapunov bound}  for expression simplicity.
Note that a solution to problem  {\bf MP'}  can be obtained by minimizing the upper bound \eqref{Eq: Lyapunov bound} in each time slot $t$, in which the optimality of  {\bf MP'} is asymptotically approached by increasing $V$  \cite{Neely/Stochastic}. 
To minimize \eqref{Eq: Lyapunov bound}, we have three decomposed optimization problems {\bf P1}, {\bf P2}, and {\bf P3} which are detailed and solved in the following parts.

The first decomposed problem, which jointly associates UEs with MEC servers and allocates UEs' computational and communication resources, is given by
%
%
%
\begin{IEEEeqnarray*}{rcl}
\mbox{\bf P1:}~~&\underset{\boldsymbol{\eta}(n), \mathbf{f}(t),\mathbf{P}(t)}{\mbox{minimize}}&~~
\sum_{i\in\mathcal{U}}\sum_{j\in\mathcal{S}}\big(\beta_{ji}(t)-\beta_{i}^{\rm O}(t)\big) \tau W \log_2\bigg(1\notag
\\&&~~+\frac{\eta_{ij}(n)P_{i}(t)h_{ij}(t)}{N_0W+\sum\limits_{i'\in\mathcal{U}\setminus i}\eta_{i'j}(n)P_{i'}(t)h_{i'j}(t)}\bigg) \notag
\\&&~~+\sum\limits_{i\in\mathcal{U}}\Big[ V\big(\kappa [f_i(t)]^3+P_{i}(t)\big)-\frac{\beta_{i}^{\rm L}(t)\tau f_{i}(t)}{L_i}\Big] \notag
\\&\mbox{subject to}&~~\mbox{\eqref{Eq: UE-server association constraint} and \eqref{Eq: UE resource constraint},}\notag
\end{IEEEeqnarray*}
%
%
%
with
\begin{align}
&\beta_{ji}(t)=\Big(Q^{(\rm X)}_{ji}(t)+2Q^{(\rm Y)}_{ji}(t)Z_{ji}(t)+2\big[Z_{ji}(t)\big]^3+Z_{ji}(t)\Big)\notag
\\&\times\mathbbm{1}_{\big\{ Z_{ji}(t)+  \tau R^{\max}_{ij}(t)>\tilde{R}_{ij}(t-1) d_{ji}\big\}}+Q^{(\rm Z)}_{ji}(t)+Z_{ji}(t),\label{Eq: weight server}
\\&\beta_{i}^{\rm O}(t)=\Big(Q^{\rm O, (X)}_{i}(t)+2Q^{\rm O, (Y)}_{i}(t)Q^{\rm O}_{i}(t)+2\big[Q^{\rm O}_{i}(t)\big]^3\notag
\\&+Q^{\rm O}_{i}(t) \Big)\times\mathbbm{1}_{\{Q^{\rm O}_{i}(t)+  A_{i}(t)> d^{\rm O}_{i}\}}+Q^{\rm O, (Q)}_{i}(t)+Q^{\rm O}_{i}(t),
\\&\beta_{i}^{\rm L}(t)=\Big(Q^{\rm L, (X)}_{i}(t)+2Q^{\rm L, (Y)}_{i}(t)Q^{\rm L}_{i}(t)+2\big[Q^{\rm L}_{i}(t)\big]^3\notag
\\&+Q^{\rm L}_{i}(t) \Big)\times\mathbbm{1}_{\{Q^{\rm L}_{i}(t)+  A_{i}(t)> d^{\rm L}_{i}\}}+Q^{\rm L, (Q)}_{i}(t)+Q^{\rm L}_{i}(t).
\end{align}
Note that in {\bf P1}, the UE's allocated transmit power is coupled with the local CPU-cycle frequency. The transmit power also depends on the wireless channel strength to the associated server and the weight $\beta_{ji}(t)$ of the corresponding offloaded-task queue, in which
the former depends on the distance between the UE and  server when the latter is related to the MEC server's computation capability and the number of associated UEs. Therefore, the UEs' geographic configuration and the servers' computation capabilities should be taken into account while we associate the UEs with the servers.
Moreover,  UE-server association, i.e., $\boldsymbol{\eta}(n)$, and resource allocation, i.e., $\mathbf{f}(t)$ and $\mathbf{P}(t)$, are performed in two different timescales, i.e.,  in the beginning of each time frame and every time slot afterwards. We solve {\bf P1} in two steps, in which the UE-server association is firstly decided. Then, given the association results, UEs' CPU-cycle frequencies and transmit powers are allocated.

\subsection{UE-Server Association using Many-to-One Matching with Externalities}

 To associate  UEs to the MEC servers, let us focus on the wireless transmission part of {\bf P1} and, thus, fix $P_{i}(t)= P_{i}^{\max}$ and  $f_i(t)=0,\forall\,i\in\mathcal{U}$, at this stage. The wireless channel gain $h_{ij}(t)$ and the weight factors $\beta_{i}^{\rm O}(t)$ and  $\beta_{ji}(t)$ dynamically change in each time slot, whereas the UEs are re-associated with the servers in every $T_0$ slots. In order to take the impacts of $h_{ij}(t)$, $\beta_{i}^{\rm O}(t)$, and  $\beta_{ji}(t)$  into account, we consider the average strength for the channel gain, i.e., letting $h_{ij}(t)=\mathbb{E}[h_{ij}],\forall\,i\in\mathcal{U},j\in\mathcal{S}$, and the empirical average, i.e.,
\begin{align*}
\tilde{\beta}_{i}^{\rm O}(n)&=\frac{1}{(n-1)T_0}\sum\limits_{t=0}^{(n-1)T_0-1}\beta_{i}^{\rm O}(t),~\forall\,i\in\mathcal{U},
\\\tilde{\beta}_{j}(n)&=\frac{1}{(n-1)T_0}\sum\limits_{k=1}^{n-1}\sum\limits_{t=(k-1)T_0}^{kT_0-1}\sum\limits_{i\in\mathcal{U}}\frac{\eta_{ij}(k)\beta_{ji}(t)}{\sum\limits_{i\in\mathcal{U}}\eta_{ij}(k)},~\forall\,j\in\mathcal{S},
\end{align*}
as the estimations of the weight factors $\beta_{i}^{\rm O}(t)$ and  $\beta_{ji}(t)$.  Here, $\tilde{\beta}_{j}(n)$ represents the estimated average weight factor of a single UE's offloaded-task queue (at server $j$) since we also take the average over all the associated UEs in each time frame.
Incorporating the above assumptions, the UE-server association problem of {\bf P1} can be considered as
%
%
%
\begin{IEEEeqnarray*}{rcl}
\mbox{\bf P1-1:}~~&\underset{\boldsymbol{\eta}(n)}{\mbox{maximize}}
&~~\sum_{i\in\mathcal{U}}\sum_{j\in\mathcal{S}}\big(\tilde{\beta}_{i}^{\rm O}(n) - \tilde{\beta}_{j}(n)\big) \log_2\bigg(1\notag
\\&&~~+\frac{\eta_{ij}(n)P_{i}^{\max}\mathbb{E}[h_{ij}]}{N_0W+\sum\limits_{i'\in\mathcal{U}\setminus i}\eta_{i'j}(n)P_{i'}^{\max}\mathbb{E}[h_{i'j}]}\bigg) 
\\&\mbox{subject to}&~~\eqref{Eq: UE-server association constraint}.
\end{IEEEeqnarray*}
%
%
However, due to the binary nature of optimization variables and the non-convexity of the objective, {\bf P1-1} is an NP-hard nonlinear integer programming problem \cite{Nonlinear_program}. To tackle this, we invoke matching theory which provides an efficient and low-complexity way to  solve integer programming problems \cite{Matching_mag,Many_to_Many_exter} and has been utilized in  various wireless communication systems \cite{Zhang_IoTJ17,Many_to_Many_exter,Matching_cache,Matching_V2V,Di_matching}.

Defined in matching theory, \emph{matching} is a bilateral assignment between two sets of players, i.e., the sets of UEs and MEC servers, which matches each player in one set to the player(s) in the opposite set \cite{matching_theo}.
Referring to the structure of problem {\bf P1-1}, we consider a many-to-one matching model  \cite{Many_one} in which each UE $i\in\mathcal{U}$ is assigned to a server  when each server $j\in\mathcal{S}$ can be assigned to multiple UEs. Moreover,  UE $i$ is assigned to server $j$ if server $j$ is assigned to UE $i$, and vice versa. The considered many-to-one matching is formally defined as follows.
\begin{definition}
The considered many-to-one matching game  consists of two sets of players, i.e., $\mathcal{U}$ and $\mathcal{S}$, and the outcome of the many-to-one matching is a function $\eta$ from $\mathcal{U}\times \mathcal{S}$ to the set of all subsets of $\mathcal{U}\times\mathcal{S}$ with 
\begin{enumerate}
\item
$|\eta(i)|=1,\forall\,i\in\mathcal{U}$, 
\item
 $|\eta(j)|\leq U,\forall\,j\in\mathcal{S}$, 
 \item
  $j=\eta(i)\Leftrightarrow i\in\eta(j)$.
  \end{enumerate}
\end{definition}
Having a matching $\eta$,  the UE-server association indicator can be specified as per, $\forall\,i\in\mathcal{U},j\in\mathcal{S}$, 
\begin{equation}\label{Eq: matching-association}
\begin{cases}
\eta_{ij}(n)=1,&\mbox{if~}j=\eta(i),
 \\\eta_{ij}(n)=0,&\mbox{otherwise}.
\end{cases}
\end{equation}
Noe that {\bf P1-1} is equivalent to a weighted sum rate maximization problem in which  $\big(\tilde{\beta}_{i}^{\rm O}(n) - \tilde{\beta}_{j}(n)\big) \log_2\Big(1+\frac{P_{i}^{\max}\mathbb{E}[h_{ij}]}{N_0W+\sum\limits_{i'\in\mathcal{U}\setminus i}\eta_{i'j}(n)P_{i'}^{\max}\mathbb{E}[h_{i'j}]}\Big)$ can be treated as UE $i$'s weighted transmission rate 
provided that UE $i$ is matched to server $j$.  Thus, the UE's matching preference over the servers can be chosen based on the  weighted rates in the descending order.
Since the rate is affected by the other UEs (i.e., the players in the same set of the matching game) via interference if they are matched to the same server, the UE's preference also depends on the matching state of the other UEs. 
The interdependency between the players' preferences is called  \emph{externalities} \cite{Many_to_Many_exter} in which  the player's preference dynamically changes with the matching state of the other players in the same set.
Thus, the UE's preference over matching states should be adopted.
To this end, given a matching $\eta$ with $j=\eta(i)$, we define UE $i$'s utility as
\begin{align}\label{Eq: UE utility}
& \Psi_i(\eta)=\big(\tilde{\beta}_{i}^{\rm O}-\tilde{\beta}_{j}\big) \log_2\bigg(1\notag
 \\&+\frac{P_{i}^{\max}\mathbb{E}[h_{ij}]}{N_0W+\sum\limits_{i'\in\mathcal{U}\setminus i}\mathbbm{1}_{\{\eta(i')=j\}}P_{i'}^{\max}\mathbb{E}[h_{i'j}]}\bigg),~\forall\,i\in\mathcal{U},
\end{align}
and server $j$'s utility as
\begin{align}\label{Eq: server utility}
& \Psi_j(\eta)=\sum_{i\in\mathcal{U}}
 \big(\tilde{\beta}_{i}^{\rm O}-\tilde{\beta}_{j}\big) \log_2\bigg(1\notag
 \\&+\frac{\eta_{ij}P_{i}^{\max}\mathbb{E}[h_{ij}]}{N_0W+\sum\limits_{i'\in\mathcal{U}\setminus i}\mathbbm{1}_{\{\eta(i')=j\}}P_{i'}^{\max}\mathbb{E}[h_{i'j}]}\bigg),~\forall\,j\in\mathcal{S}.
\end{align}
The UE's and server's matching preferences are based on their own utilities in a descending order.  For notational simplicity. we remove the time index $n$ in \eqref{Eq: UE utility} and \eqref{Eq: server utility}.
Subsequently, we consider the notion of \emph{swap matching} to deal with externalities \cite{Many_one}.
\begin{definition}\label{Def: swap}
Given a may-to-one matching $\eta$, a pair of UEs $(i,i')$, and a pair of servers $(j,j')$ with $j=\eta(i)$ and $j'=\eta(i')$,
a swap matching $\eta_{ij}^{i'j'}$ is
\begin{align*}
\eta_{ij}^{i'j'}=\big\{\eta\setminus \{(i,j), (i',j')\}\big\}\cup\big\{(i,j'),(i',j)\big\}.
\end{align*}
\end{definition}
In other words, we have $j'=\eta_{ij}^{i'j'}(i)$ and $j=\eta_{ij}^{i'j'}(i')$  in the swap matching $\eta_{ij}^{i'j'}$ for the UE pair $(i,i')$ and server pair $(j,j')$, whereas the matching state of the other UEs and servers remains identical in both $\eta$ and $\eta_{ij}^{i'j'}$.
Furthermore, in Definition \ref{Def: swap}, one of the UE pair in the swap operation, e.g., $i'$, can be an open spot of server $j'$ in $\eta$ with $|\eta(j')|<U$. In this situation, we have $|\eta_{ij}^{i'j'}(j')|-|\eta(j')|=1$ and $|\eta(j)|-|\eta_{ij}^{i'j'}(j)|=1$. Moreover,  $\Psi_{i'}(\eta)=0$ and  $\Psi_{i'}(\eta_{ij}^{i'j'})=0$.
\begin{definition}
For the matching $\eta$, $(i,i')$ is a swap-blocking pair \cite{Many_to_Many_exter} if and only if 
\begin{enumerate}
\item
 $\forall\,u\in\{i,i',j,j'\}$, $\Psi_u(\eta_{ij}^{i'j'})\geq \Psi_u(\eta)$,
 \item
  $\exists\,u\in\{i,i',j,j'\}$ such that $\Psi_u(\eta_{ij}^{i'j'})> \Psi_u(\eta)$.
  \end{enumerate}
\end{definition}
Therefore, provided that the matching state of the remaining UEs and servers is fixed, two UEs $(i,i')$ exchange their respectively matched servers if both UEs' and both servers' utilities will not be worse off, and at lest one's utility is better off, after the swap. The same criteria are applicable when the UE $i$ changes to another server $j'$ from the current matched server $j$, i.e., $i'$ is an open spot of server $j'$.
\begin{definition}\label{Def: Two-side stable}
A matching $\eta$ is two-sided exchange-stable if there is no  swap-blocking pair \cite{Many_one}.
\end{definition}

In summary,  we first initialize a matching $\eta$ and calculate utilities \eqref{Eq: UE utility} and \eqref{Eq: server utility}. Then, we iteratively find a swap-block pair and update the swap matching until two-sided exchange stability is achieved. The steps to solve problem {\bf P1-1} by many-to-one matching with externalities are detailed in Algorithm \ref{Alg: matching}. 
In each iteration of Algorithm \ref{Alg: matching}, there are $\binom U2+U(S-1)$ possibilities to form a swap-blocking pair, in which there are $\binom U2$ pairs of two different UEs.  The term $U(S-1)$ accounts for all the swap-blocking pairs consisting of a UE $i\in\mathcal{U}$ and an open spot of another MEC server $j'\in\mathcal{S}\setminus j$.
Given that the UEs are more than the MEC servers in our considered network, the complexity of Algorithm \ref{Alg: matching} is in the order of $\mathcal{O}(U^2)$ which increases binomially with the number of UEs.
Additionally, let us consider an exhaustive search in problem {\bf P1-1}. 
Since each UE $i\in\mathcal{U}$ can only access one out of $S$ servers, there are  $S^U$ association choices satisfying constraint \eqref{Eq: UE-server association constraint}.
In the exhaustive search approach, the complexity increases exponentially with the number of UEs.

\subsection{Resource Allocation and Task Splitting at the UE Side}

Now denoting the representative UE $i$'s associated MEC server as $j^*$,  we formulate the UEs' resource allocation problem of {\bf P1}  as
%
%
%
\begin{IEEEeqnarray*}{rcl}
\mbox{\bf P1-2:}~~&\underset{\mathbf{f}(t),\mathbf{P}(t)}{\mbox{minimize}}&~~
\sum_{i\in\mathcal{U}}\bigg[\big(\beta_{j^{*}i}(t)-\beta_{i}^{\rm O}(t)\big) \tau W \log_2\bigg(1\notag
\\&&~~+\frac{P_{i}(t)h_{ij^{*}}(t)}{N_0W+\sum\limits_{i'\in\mathcal{U}\setminus i}\eta_{i'j^{*}}P_{i'}(t)h_{i'j^{*}}(t)}\bigg) \notag
\\&&~~- \frac{\beta_{i}^{\rm L}(t)\tau f_{i}(t)}{L_i}+ V\big(\kappa [f_i(t)]^3+P_{i}(t)\big)\bigg]
\\&\mbox{subject to}&~~\eqref{Eq: UE resource constraint}.\notag
\end{IEEEeqnarray*}
%
%
%
%
%
%
%
When solving problem {\bf P1-2} in a centralized manner,  each UE $i$ needs to upload its local information $\beta_{i}^{\rm L}(t)$ and $\beta_{i}^{\rm O}(t)$  in every time slot to a central unit, e.g., the associated server.  This can incur high overheads, especially in dense networks.
In order to alleviate this issue, we decompose the summation (over all UEs) in the objective and let each UE $i\in\mathcal{U}$  locally allocate its CPU-cycle frequency and transmit power by solving
%
%
%
\begin{IEEEeqnarray*}{rcl}
\mbox{\bf P1-2':}~~&\underset{f_i(t),P_i(t)}{\mbox{minimize}}&~~
\big(\beta_{j^{*}i}(t)-\beta_{i}^{\rm O}(t)\big)\tau W\notag
\\& &~~\times\mathbb{E}_{I_{ij^{*}}}\bigg[\log_2\bigg(1+\frac{P_{i}(t)h_{ij^{*}}(t)}{N_0W+I_{ij^{*}}(t)}\bigg)\bigg]\notag
 \\ &&~~ - \frac{\beta_{i}^{\rm L}(t)\tau f_{i}(t)}{L_i}+V\big(\kappa [f_i(t)]^3+P_{i}(t)\big)
%
%
\\&\mbox{subject to}&~~\eqref{Eq: UE resource constraint}.
\end{IEEEeqnarray*}
%
%
%
The expectation is with respect to the locally-estimated distribution, $\hat{\Pr}\big(I_{ij};t\big)$,  of the aggregate interference  $I_{ij^*}(t)=\sum_{i'\in\mathcal{U}\setminus i}\eta_{i'j^{*}}P_{i'}(t)h_{i'j^{*}}(t)$. 
Note that when $V=0$, {\bf P1-2'} is equivalent to the rate maximization problem, where the power budget will be fully allocated. As $V$ is gradually increased, we pay more attention on power cost reduction, and the solution values will decrease correspondingly.
Moreover, although the downlink is not considered in this work, we implicitly assume that the UE has those executed tasks and can locally track the state of the offloaded-task queue. In other words,   the full information about $\beta_{j^{*}i}(t)$ is available at the UE.
 \begin{lemma}\label{lemma power}
 
The optimal solution to problem {\bf P1-2'} is that
UE $i$ allocates the CPU-cycle frequency
\begin{align*}
 f^{*}_i(t)=\sqrt{\frac{ \beta_{i}^{\rm L}(t)\tau }{ 3L_i\kappa(V+\gamma^*)}}
\end{align*}
for local computation. The optimal allocated transmit  power $P^{*}_{i}>0$ satisfies
\begin{align}\notag
  \mathbb{E}_{I_{ij^{*}}}\bigg[\frac{\big(\beta_{i}^{\rm O}(t)-\beta_{j^{*}i}(t)\big)\tau W h_{ij^{*}}}{(N_0W+I_{ij^{*}}+P^{*}_{i}h_{ij^{*}})\ln 2}\bigg]
 =V+\gamma^{*}
\end{align}
if $  \mathbb{E}_{I_{ij^{*}}}\Big[\frac{(\beta_{i}^{\rm O}(t)-\beta_{j^{*}i}(t))\tau W h_{ij^{*}}}{(N_0W+I_{ij^{*}})\ln 2}\Big]>V+\gamma^*$. Otherwise,  $P^{*}_{i}=0$. 
Furthermore, the optimal Lagrange multiplier $\gamma^*$ is 0 if $\kappa [f^*_i(t)]^3+P^*_{i}(t)< P_{i}^{\max}$. When $\gamma^*> 0$,
$\kappa [f^*_i(t)]^3+P^*_{i}(t)= P_{i}^{\max}$.
\end{lemma}
\begin{proof}
Please refer to Appendix \ref{Lem: opt pow}.
\end{proof}

\begin{algorithm}[t]
  \caption{UE-Server Association by Many-to-One Matching with Externalities}
  \begin{algorithmic}[1]
    \State Initialize $\eta(i)=\operatorname{argmax}_{j\in\mathcal{S}} \big\{\mathbb{E}[h_{ij}] \big\},\forall\,i\in\mathcal{U}$.
    \State  Calculate \eqref{Eq: UE utility} and \eqref{Eq: server utility}.
    \Repeat 
 \State   Select  a pair of UEs $(i,i')$  or a UE $i$ with an open spot $i'$ of server $j'$. 
      \If{$(i,i')$ is a swap-blocking pair of the current matching $\eta$}
      \State Update $\eta\gets\eta_{ij}^{i'j'} $.
      \State   Calculate \eqref{Eq: UE utility} and \eqref{Eq: server utility}.
    \EndIf
    \Until{No swap-blocking pair exists  in the current matching $\eta$.}
    \State Transfer  $\eta$ to the UE-server association indicator  as per \eqref{Eq: matching-association}.
  \end{algorithmic}
  \label{Alg: matching}
\end{algorithm}

The second decomposed problem {\bf P2} given in  \eqref{Eq: Lyapunov bound} determines whether the arrival task is assigned to the local-computation queue or task-offloading queue. In this regard, each UE $i\in\mathcal{U}$ solves
%
\begin{IEEEeqnarray*}{rcl}
\mbox{\bf P2:}~~&\underset{A^{\rm L}_{i}(t), A^{\rm O}_{i}(t)}{\mbox{minimize}}&~~\beta_{i}^{\rm L}(t)
  A^{\rm L}_{i}(t)+\beta_{i}^{\rm O}(t)A^{\rm O}_{i}(t)
\\&\mbox{subject to}&~~A_i(t)=A^{\rm L}_{i}(t)+ A^{\rm O}_{i}(t),
\\&&~~A^{\rm L}_{i}(t),A^{\rm O}_{i}(t)\in\{0,A_{\rm unit},2A_{\rm unit},\cdots\}.
\end{IEEEeqnarray*}
%
%
%
%
%
We can straightforwardly find  an optimal solution $(A^{\rm L *}_{i}(t),A^{\rm O *}_{i}(t))$ as
\begin{equation}\label{Eq: Task split rule}
\begin{cases}
\big(A^{\rm L *}_{i}(t),A^{\rm O *}_{i}(t)\big)=(A_i(t),0),&\mbox{if }\beta_{i}^{\rm L}(t)\leq\beta_{i}^{\rm O}(t),
\\\big(A^{\rm L *}_{i}(t),A^{\rm O *}_{i}(t)\big)=(0,A_i(t)),&\mbox{if }\beta_{i}^{\rm L}(t)>\beta_{i}^{\rm O}(t).
\end{cases}
\end{equation}

\subsection{Computational Resource Scheduling at the Server Side}\label{Sec: Server resource scheduling}
The third decomposed problem {\bf P3} given in \eqref{Eq: Lyapunov bound} is addressed at the server side, where each MEC server $j\in\mathcal{S}$ solves
%
%
%
%
\begin{IEEEeqnarray*}{rcl}
\mbox{\bf P3:}~~&\underset{\mathbf{f}_j(t)}{\mbox{maximize}}&~~\sum_{i\in\mathcal{U}}\frac{ \beta_{ji}(t) f_{ji}(t)}{L_i} \label{Eq: server problem-1}
\\&\mbox{subject to}&~~\eqref{Eq: Server resource constraint},
\end{IEEEeqnarray*}
%
%
%
to schedule its CPU cores. 
Note that the server considers the allocated CPU-cycle frequencies to all UEs  in the objective, constraints, and optimization variables even though some UEs are not associated in the current time frame $n$. However, since  the corresponding weight of the UE's offloaded-task queue, i.e., $\beta_{ji}(t)$, is taken into account in the objective, the resource allocation at the server will not be over-provisioned. These insights are illustrated as follows.
\begin{algorithm}[t]
  \caption{Computational Resource Scheduling at the Server}
  \begin{algorithmic}[1]
    \State Initialize $n=1$, $\mathcal{U}_j=\{i\in\mathcal{U}|\beta_{ji}(t)>0\}$, and $f_{ji}=0,\forall\,i\in\mathcal{U}$.
    \While{$n\leq  N_j$ and $\mathcal{U}_j\neq \emptyset$}
      \State Find $i^*=\operatorname{argmax}_{i\in\mathcal{U}_j} \big\{\beta_{ji}(t)/L_i \big\}$.
      \State Let $f_{ji^{*}}(t)=f_{j}^{\max}$.
      \State Update $n\gets n+1$ and $\mathcal{U}_j\gets\mathcal{U}_j\setminus i^{*}$.
      \EndWhile
  \end{algorithmic}\label{Alg: cloud}
\end{algorithm}

Assuming that the server is not able to complete a UE's  offloaded tasks in the previous time frame $n-1$ and is not associated with  this UE in the current time frame $n$. Although the offloaded-task queue length $Z_{ji}(t)$ does not grow anymore, those incomplete tasks will experience severe delay if they are ignored in the current time frame. Furthermore,  since $\tilde{R}_{ij}(t)$ that is considered in the URLLC constraints decreases in the current frame,  the weight $\beta_{ji}(t)$  grows as per  \eqref{Eq: virtual GPD_server-1}, \eqref{Eq: virtual GPD_server-2},  \eqref{Eq: virtual violation_server-1}, and  \eqref{Eq: weight server}. In other words, the more severe the experienced delay of the ignored UE's incomplete tasks, the higher the weight $\beta_{ji}(t)$.
To address this severe-delay issue, the server takes into account all UEs via $\beta_{ji}(t)$ in the objective. Once the offloaded tasks are completed,  $\beta_{ji}(t)$ remains zero in the rest time slots. In addition, for the UEs which have not been associated with this server, we have $\beta_{ji}(t)=0$. Therefore, the server only considers the UEs with $\beta_{ji}(t)>0$ while scheduling the computational resources.
In summary, the optimal solution to problem {\bf P3} is that server $j$ dedicates its CPU cores to, at most,  $N_j$ UEs with largest positive values of $b_{j i}(t)/L_i$. 
The steps of allocating the server's CPU cores are detailed in Algorithm \ref{Alg: cloud}.

 \begin{algorithm}[t]
  \caption{Two-Timescale Mechanism for UE-Server Association, Task Offloading, and Resource Allocation}
  \begin{algorithmic}[1]
    \State Initialize $t=0$, predetermine the system lifetime as $T$, and set the initial queue values of \eqref{Eq: local task queue}, \eqref{Eq: offloading task queue},  \eqref{Eq: computation queue}, and \eqref{Eq: virtual GPD_Loc-1}--\eqref{Eq: virtual GPD_server-2} and \eqref{Eq: virtual violation_Loc-1}--\eqref{Eq: virtual violation_server-1} as zero.
      \Repeat
      \If{$t/T_0\in\mathbb{N}$}
      \State Run Algorithm \ref{Alg: matching} to associate the UEs with the MEC servers.
    \EndIf
    \State The UE allocates the local CPU-cycle frequency and transmit power as per Lemma \ref{lemma power}, and splits the task arrivals for local computation and offloading according to \eqref{Eq: Task split rule}.
    \State The server schedules its computational resources by following Algorithm \ref{Alg: cloud}.
    \State The UE updates the queue lengths in  \eqref{Eq: local task queue}, \eqref{Eq: offloading task queue}, \eqref{Eq: virtual GPD_Loc-1}--\eqref{Eq: virtual GPD_Off-2}, \eqref{Eq: virtual violation_Loc-1}, and \eqref{Eq: virtual violation_Off-1}.
      \State The server updates  the queue lengths in  \eqref{Eq: computation queue}, \eqref{Eq: virtual GPD_server-1}, \eqref{Eq: virtual GPD_server-2}, and \eqref{Eq: virtual violation_server-1}.
\State    Update $t\gets t+1$.
\Until{$t>T$}
  \end{algorithmic}\label{Alg: Two timescale}
\end{algorithm}

After computing and offloading the tasks in time slot $t$, each UE updates its physical and virtual queue lengths in  \eqref{Eq: local task queue}, \eqref{Eq: offloading task queue}, \eqref{Eq: virtual GPD_Loc-1}--\eqref{Eq: virtual GPD_Off-2}, \eqref{Eq: virtual violation_Loc-1}, and \eqref{Eq: virtual violation_Off-1}
 when the MEC servers update \eqref{Eq: computation queue}, \eqref{Eq: virtual GPD_server-1}, \eqref{Eq: virtual GPD_server-2}, and \eqref{Eq: virtual violation_server-1} for the next slot $t+1$. Moreover, based on the transmission rate $R_{ij}(t)$, the UE empirically estimates the statistics of $I_{ij}$ for slot $t+1$ as per $ \hat{\Pr}\big(\tilde{I}_{ij};t+1\big)= \frac{\mathbbm{1}_{\{I_{ij}(t)=\tilde{I}_{ij}\}}}{t+2}+ \frac{(t+1)\hat{\Pr}(\tilde{I}_{ij};t)}{t+2}$.
The procedures of the proposed two-timescale mechanism are outlined in Algorithm \ref{Alg: Two timescale}.

\section{Numerical Results}\label{Sec: results}

\begin{table}[t]
\caption{Simulation Parameters \cite{FogPower,JSAC_queue,Gilsoo_ICC17,Al16,Mie2010}}\label{Tab: parameters}
\centering
 \begin{tabular}{|m{1.2cm}|m{3.2cm}||m{1.2cm}|m{1.4cm}|}
\hline
{\bf Parameter} &{\bf  Value }& {\bf Parameter} &{\bf  Value } \\
\hline
  \hline
  $T_0$ &100&$A_{\rm unit}$&1500 bytes\\
  \hline
$\kappa $& $10^{-27}\,\mbox{Watt}\cdot\mbox{s}^3/\mbox{cycle}^3$&$W$&10\,MHz\\
  \hline
 $N_0$&-174\,dBm/Hz&$P_{i}^{\max}$&30\,dBm\\
\hline
 $L_i$&$[1\times 10^3,4\times 10^4]\,\mbox{cycle/bit}$&  $V$&0\\
\hline
$\lambda_i$&$[10,150]$\,kbps&$f_{j}^{\max}$&$10^{10}$\,cycle/s\\
\hline
 $d^{\rm L}_{i}$&$100\tilde{A}_{i}^{\rm L}(t-1)$ (bit)&$\epsilon^{\rm L}_{i}$&0.01\\
\hline
$\sigma_i^{\rm L, th}$&40 (Mbit)&$\xi_i^{\rm L, th}$&0.3\\
 \hline
 $d^{\rm O}_{i}$&$100\tilde{A}_{i}^{\rm O}(t-1)$ (bit)&$\epsilon^{\rm O}_{i}$&0.01\\
\hline
$\sigma_i^{\rm O, th}$&40 (Mbit)&$\xi_i^{\rm O, th}$&0.3\\
\hline
$d_{ji}$&20\,sec& $\epsilon_{ji}$&0.01\\
\hline
  $\sigma_{ji}^{\rm th}$&40 (Mbit)&  $\xi_{ji}^{\rm th}$&0.3\\
  \hline
\end{tabular}
\end {table}
We consider an  indoor $100\times 100\,\mbox{m}^2$ area in which four MEC servers are deployed at the centers of  four equal-size 
quadrants, respectively. Each server is equipped with eight CPU cores. Additionally, multiple UEs, ranging from 30 to 80, are randomly distributed. 
For task offloading, we assume that the transmission frequency is 5.8\,GHz with the path loss model $24\log x+20\log 5.8+60$ (dB) \cite{rpt:itu_indoor}, where $x$ in meters is the distance between any UE and server.   Further,  all wireless channels experience Rayleigh fading with unit variance. Coherence time is 40\,ms \cite{coherence_time}. Moreover, we consider Poisson processes for task arrivals. 
The rest parameters for all UEs  $i\in\mathcal{U}$ and servers $j\in\mathcal{S}$ are listed in Table \ref{Tab: parameters}.\footnote{$\tilde{A}_{i}^{\rm L}(t-1)=\frac{1}{t}\sum\limits_{\tau=0}^{t-1}A_{i}^{\rm L}(\tau)$ and $\tilde{A}_{i}^{\rm O}(t-1)=\frac{1}{t}\sum\limits_{\tau=0}^{t-1}A_{i}^{\rm O}(\tau)$.}
\begin{figure}[t]
\centering
\subfigure[Tail distributions of the excess value of a UE's  task-offloading queue and the approximated GPD of exceedances.]{\includegraphics[width=\columnwidth]{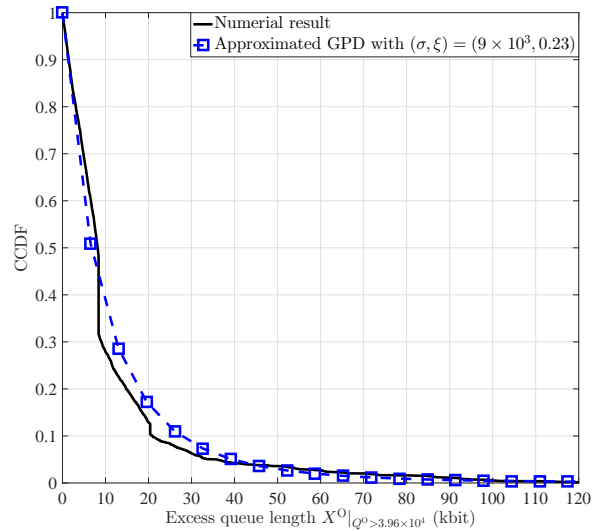}}
\subfigure[Convergence of the approximated GPD scale and shape parameters of exceedances.]{\includegraphics[width=\columnwidth]{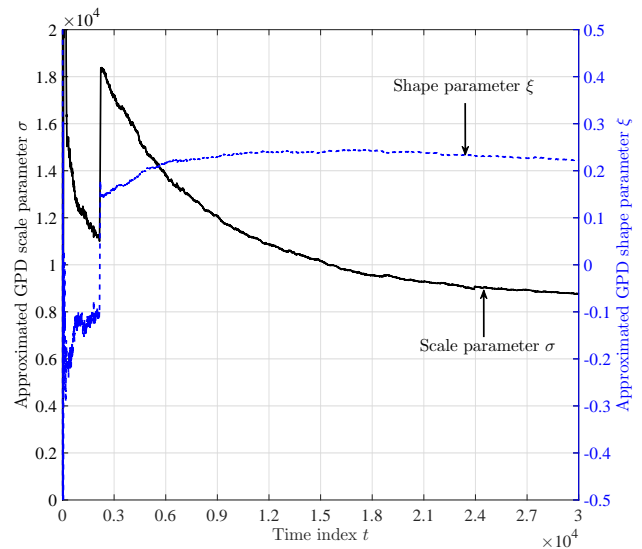}}
	\caption{Effectiveness of applying the Pickands–Balkema–de Haan theorem in the MEC network.}
\label{Fig3}
\end{figure}

We first verify the effectiveness of using the Pickands–Balkema–de Haan theorem to characterize the excess queue length in Fig.~\ref{Fig3}. Although the task-offloading queue is used to verify the results, the local-computation and offloaded-task queues can be used. 
Let us consider   $L_i=8250$\,cycle/bit, $\lambda_i=100$\,kbps, and  $d_i^{\rm O} = 10\tilde{A}^{\rm O}_i(t-1)$\,bit  for all  30 UEs. 
Then, given $\Pr\big(Q^{\rm O}>10\tilde{A}^{\rm O}(\infty)\big)=3.4\times 10^{-3}$ with $10\tilde{A}^{\rm O}(\infty)=3.96\times 10^4$,
Fig.~\ref{Fig3}(a) shows the CCDFs of exceedances $X^{\rm O}|_{Q^{\rm O}>3.96\times 10^4}=Q^{\rm O}-3.96\times 10^4>0$ and the approximated GPD in which the latter provides a good characterization for exceedances. Further, the convergence of the scale and shape parameters of the approximated GPD is shown in Fig.~\ref{Fig3}(b). Once convergence is achieved, characterizing the statistics of exceedances helps to locally estimate the network-wide extreme metrics, e.g., the maximal queue length among all UEs as in \cite{EVT_converge}, and enables us to proactively deal with the occurrence of extreme events.

\begin{figure}[t]
\centering
\includegraphics[width=\columnwidth]{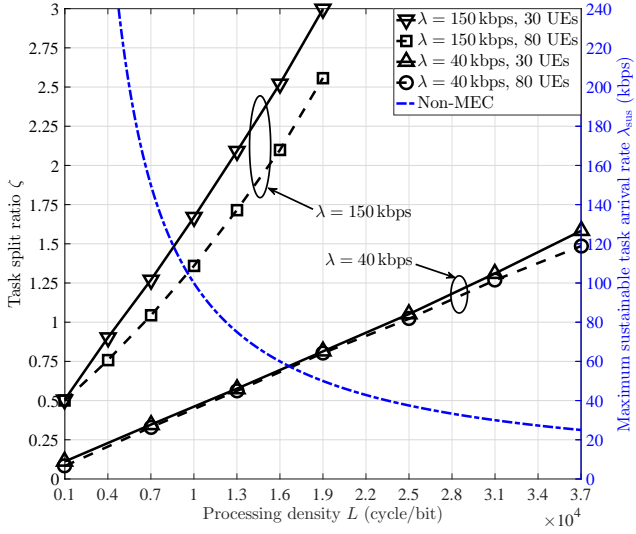}
	\caption{Task split ratio versus processing density.}
\label{Fig:4}
\end{figure}
\begin{figure}[!h]
\centering
\includegraphics[width=\columnwidth]{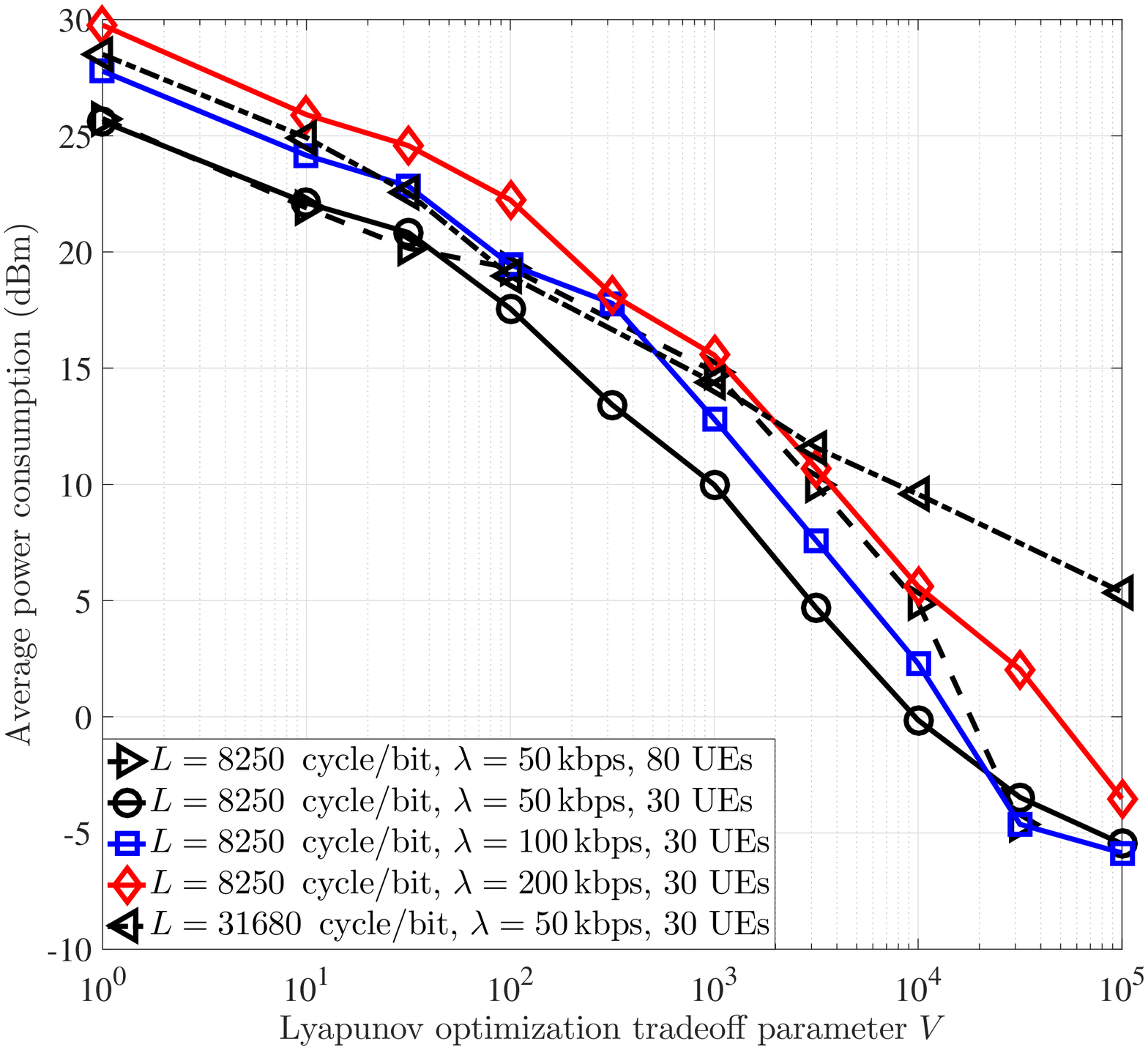}
\caption{Average power consumption versus Lyapunov  optimization tradeoff parameter.}
	\label{Fig:5}
\end{figure}
Subsequently, we measure the ratio between the task amounts split to the task-offloading queue and local-computation queues, i.e., $\zeta=\frac{\tilde{A}_{i}^{\rm O}(\infty)}{\tilde{A}_{i}^{\rm L}(\infty)}$. If $\zeta<1$, the UE pays more attention to local computation. More tasks are offloaded when $\zeta>1$.
As shown in Fig.~\ref{Fig:4}, more fraction of arrival tasks is offloaded for the intenser processing density $L$ or higher task arrival rate $\lambda$. In these situations, the UE's computation capability  becomes less supportable, and extra computational resources are required  as expected.
Additionally, since stronger interference is incurred  in denser networks, the UE lowers the task portion of offloading, especially when more  computational resources of the server are required, i.e., the cases of higher  $L$ and $\lambda$.
In the scenario without MEC servers, the local computation rate cannot be less than the task arrival rate to maintain queue stability, i.e., $10^9/L\geq \lambda$. In this regard, we also plot the curve of  the maximum sustainable task arrival rate $\lambda_{\rm sus}=10^9/L$ without MEC servers in Fig.~\ref{Fig:4}. 
Comparing the curves of the MEC schemes with the maximum sustainable task arrival rate, we can find
\begin{equation*}
\begin{cases}
 \zeta>1,&\mbox{when~} \lambda>10^9/L,
\\  \zeta=1,&\mbox{when~} \lambda=10^9/L,
  \\ \zeta<1,&\mbox{otherwise.}
 \end{cases}
\end{equation*}
That is, $\lambda_{\rm sus}$ is the watershed between task offloading and local computation. More than 50\% of arrival tasks are offloaded if the task arrival rate is higher than $\lambda_{\rm sus}$.

\begin{figure}[t]
\centering
\subfigure[ $L=8250$\,cycle/bit and 30 UEs.]{\includegraphics[width=\columnwidth]{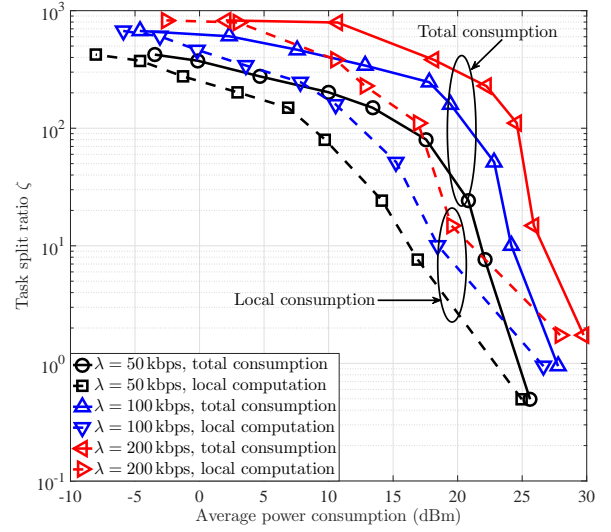}}
\subfigure[$\lambda=50$\,kbps]{\includegraphics[width=\columnwidth]{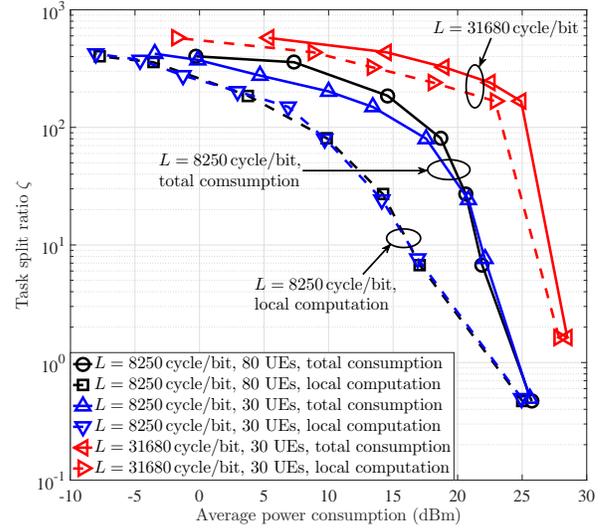}}
	\caption{Task split ratio versus average power costs of local computation and total consumption.}
	\label{Fig:6}
\end{figure}
\begin{figure*}[t]
\centering
\subfigure[Average end-to-end delay.]{\includegraphics[width=\columnwidth]{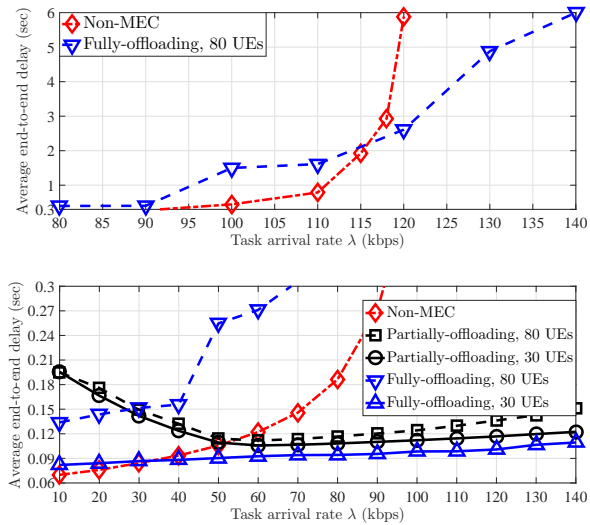}}
\subfigure[99th percentile of the UE's queue length with 80 UEs.]{\includegraphics[width=\columnwidth]{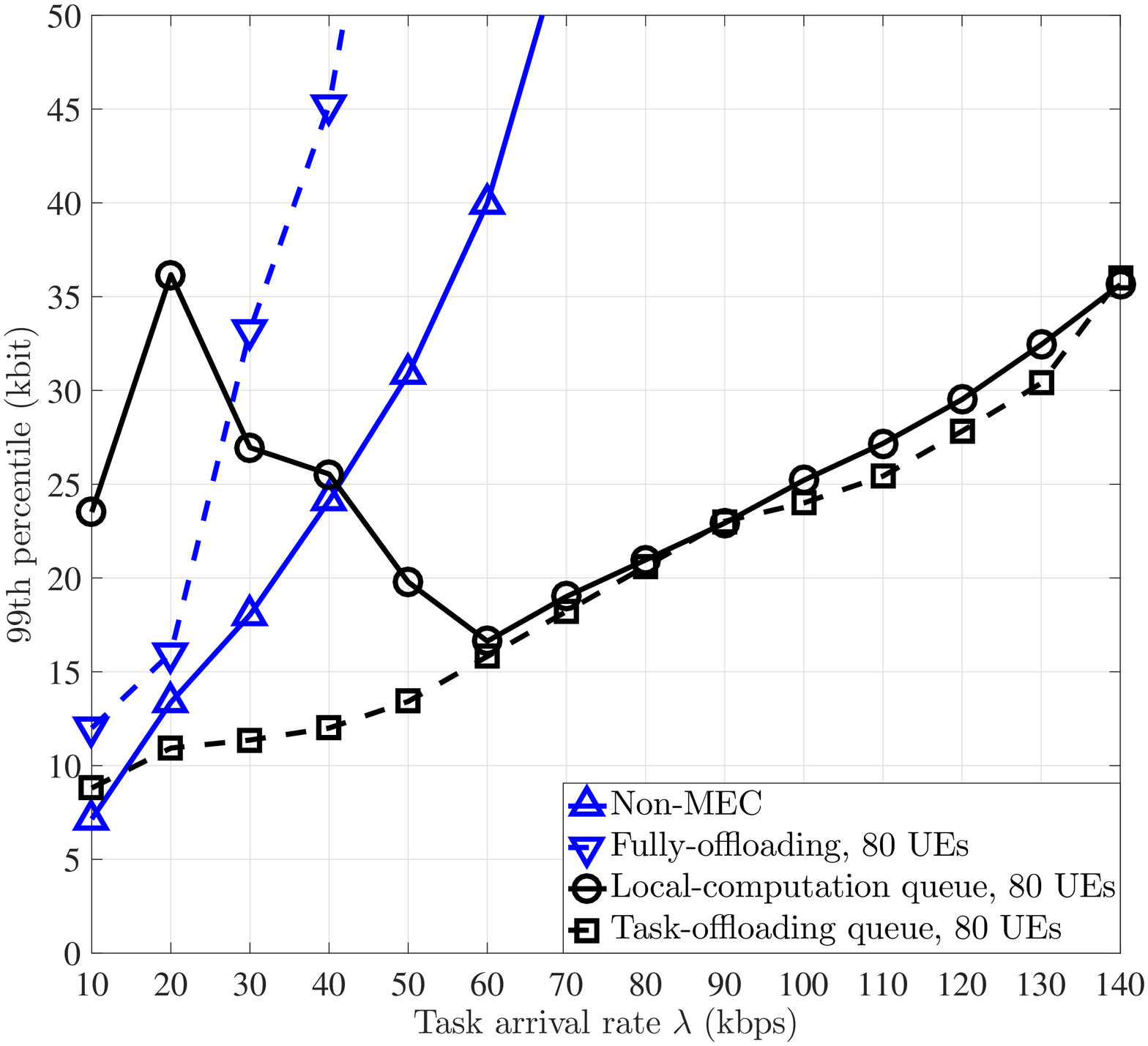}}
\subfigure[99th percentile of the UE's queue length with 30 UEs.]{\includegraphics[width=\columnwidth]{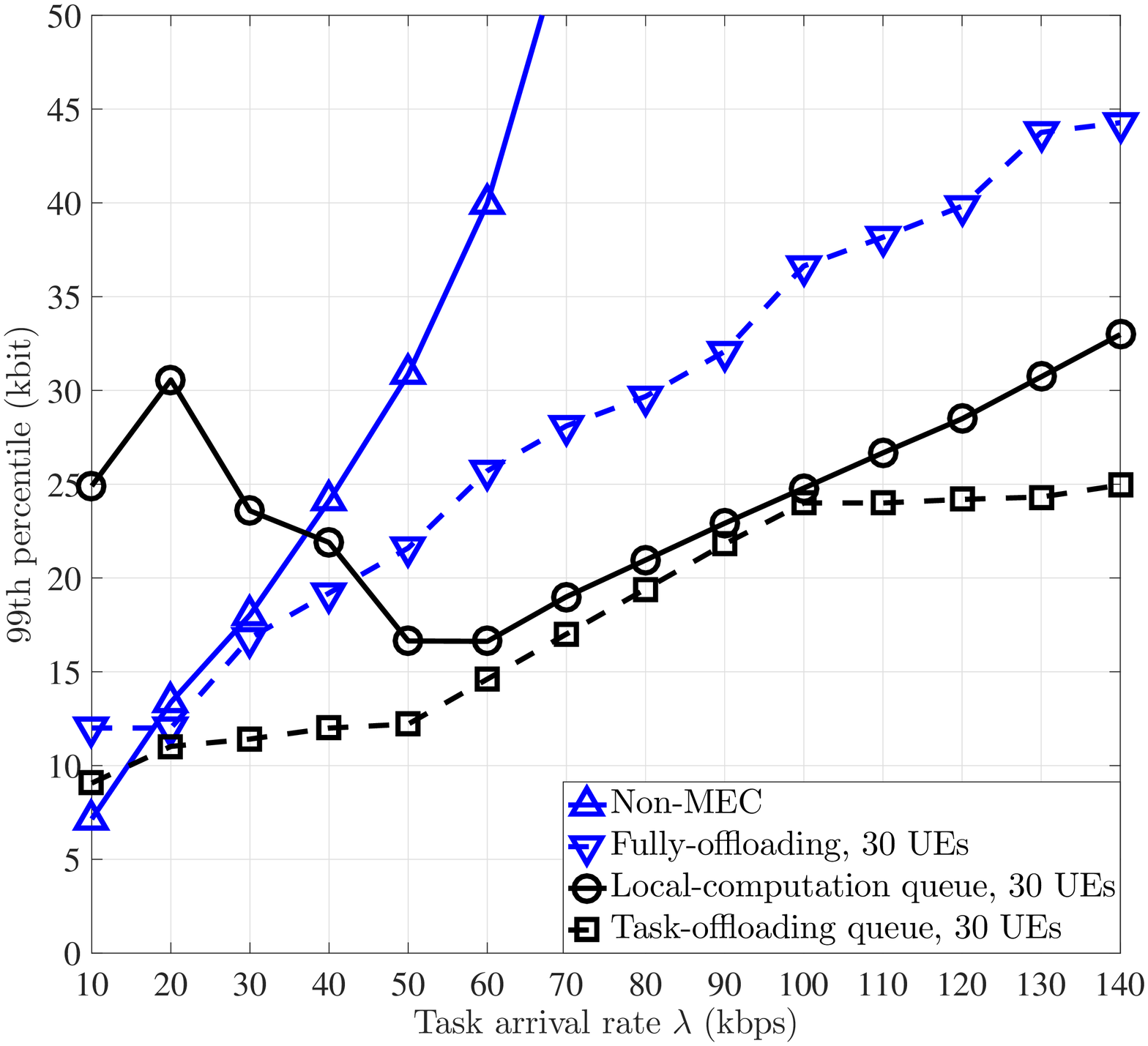}}
\subfigure[Mean and standard deviation of exceedances with 30 UEs.]{\includegraphics[width=\columnwidth]{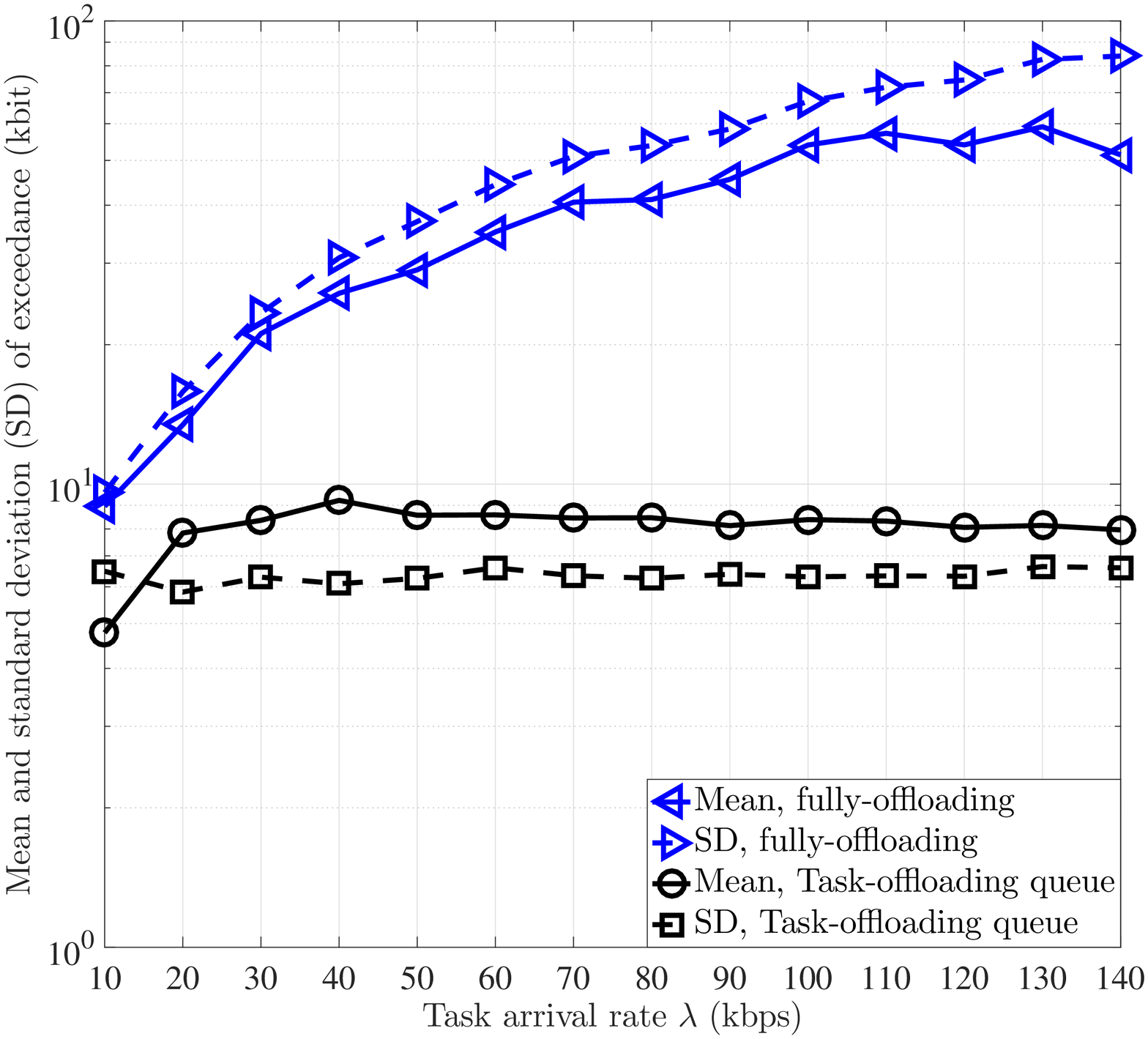}}
\caption{1) Average end-to-end delay,  2) 99th percentile of  the UE's queue length, and 3) mean and standard deviation of exceedances over the 99th percentile queue length,  versus task arrival rate with $L=8250$\,cycle/bit.}
			\label{Fig:7}
\end{figure*}
\begin{figure*}[t]
\centering
\subfigure[Average end-to-end delay.]{\includegraphics[width=\columnwidth]{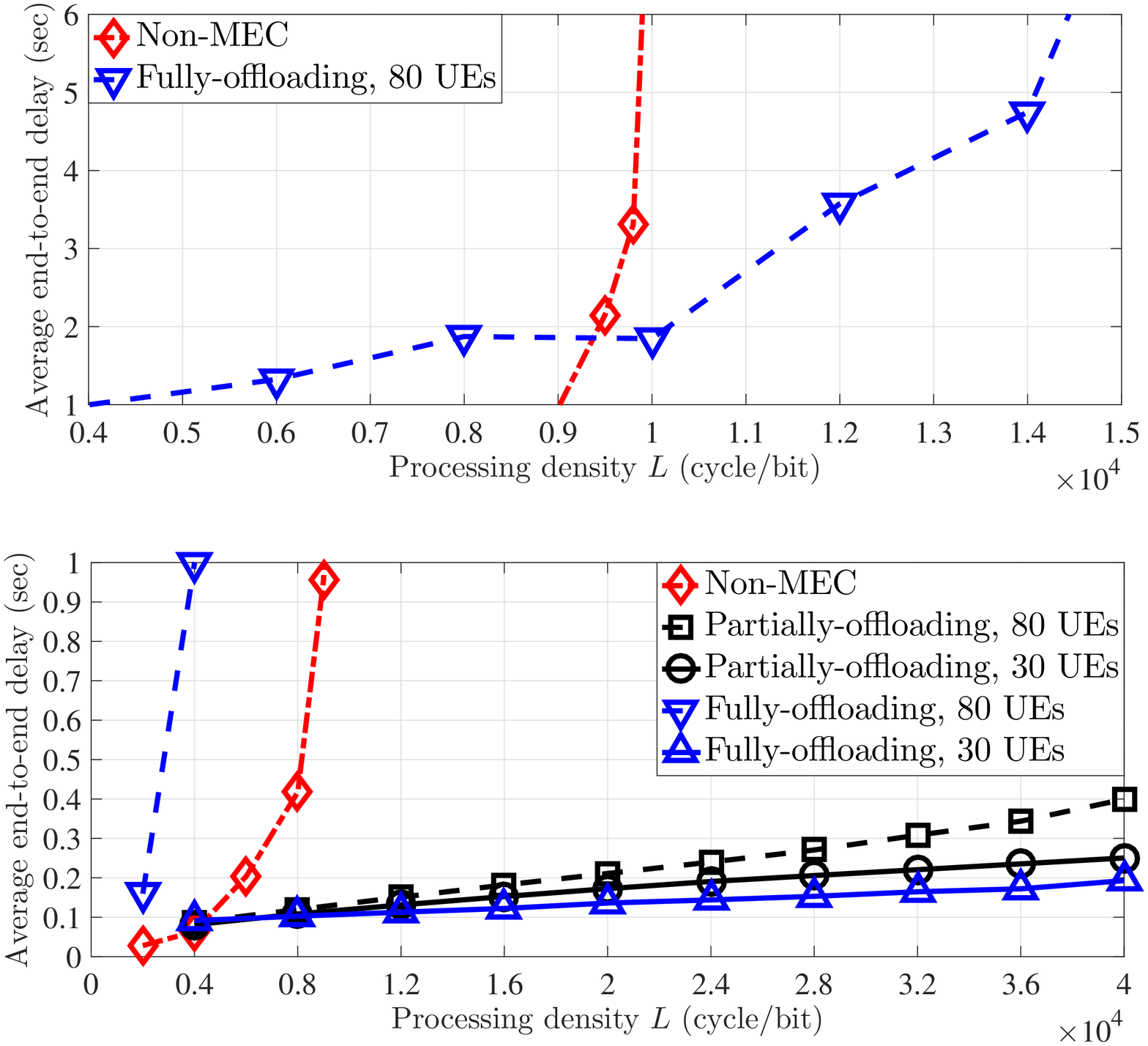}}
\subfigure[99th percentile of the UE's queue length with 80 UEs.]{\includegraphics[width=\columnwidth]{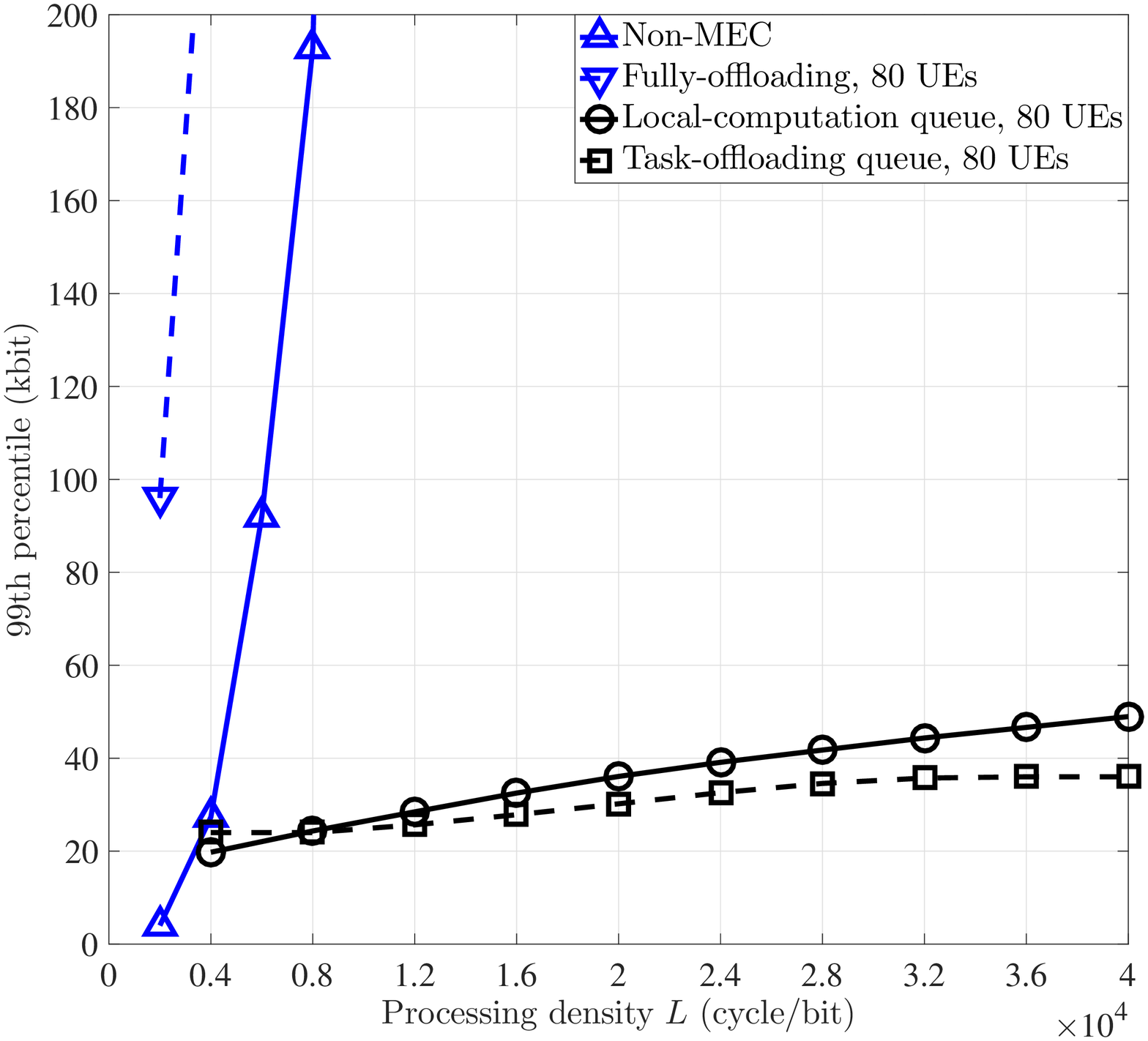}}
\subfigure[99th percentile of the UE's queue length with 30 UEs.]{\includegraphics[width=\columnwidth]{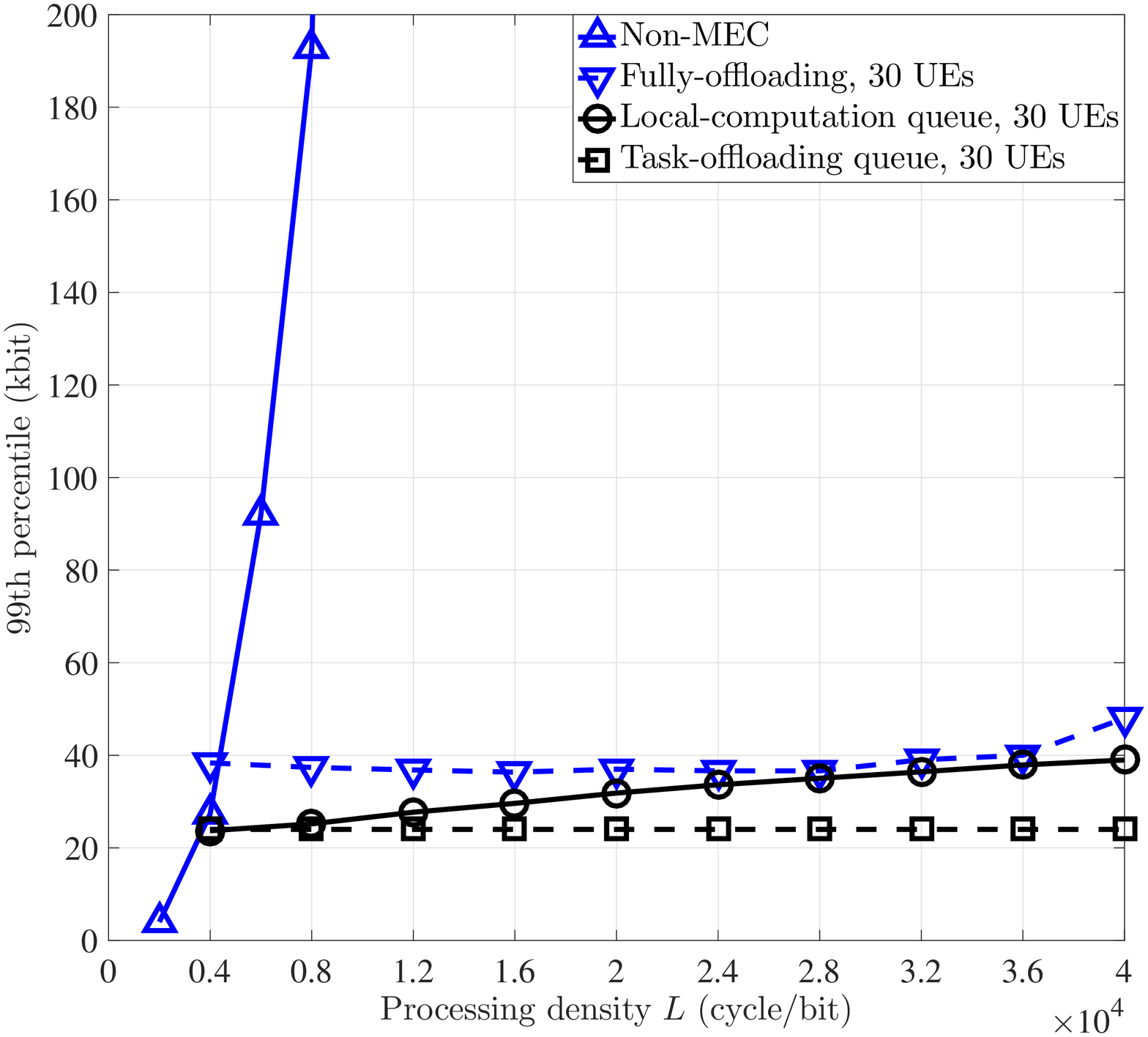}}
\subfigure[Mean and standard deviation of exceedances with 30 UEs.]{\includegraphics[width=\columnwidth]{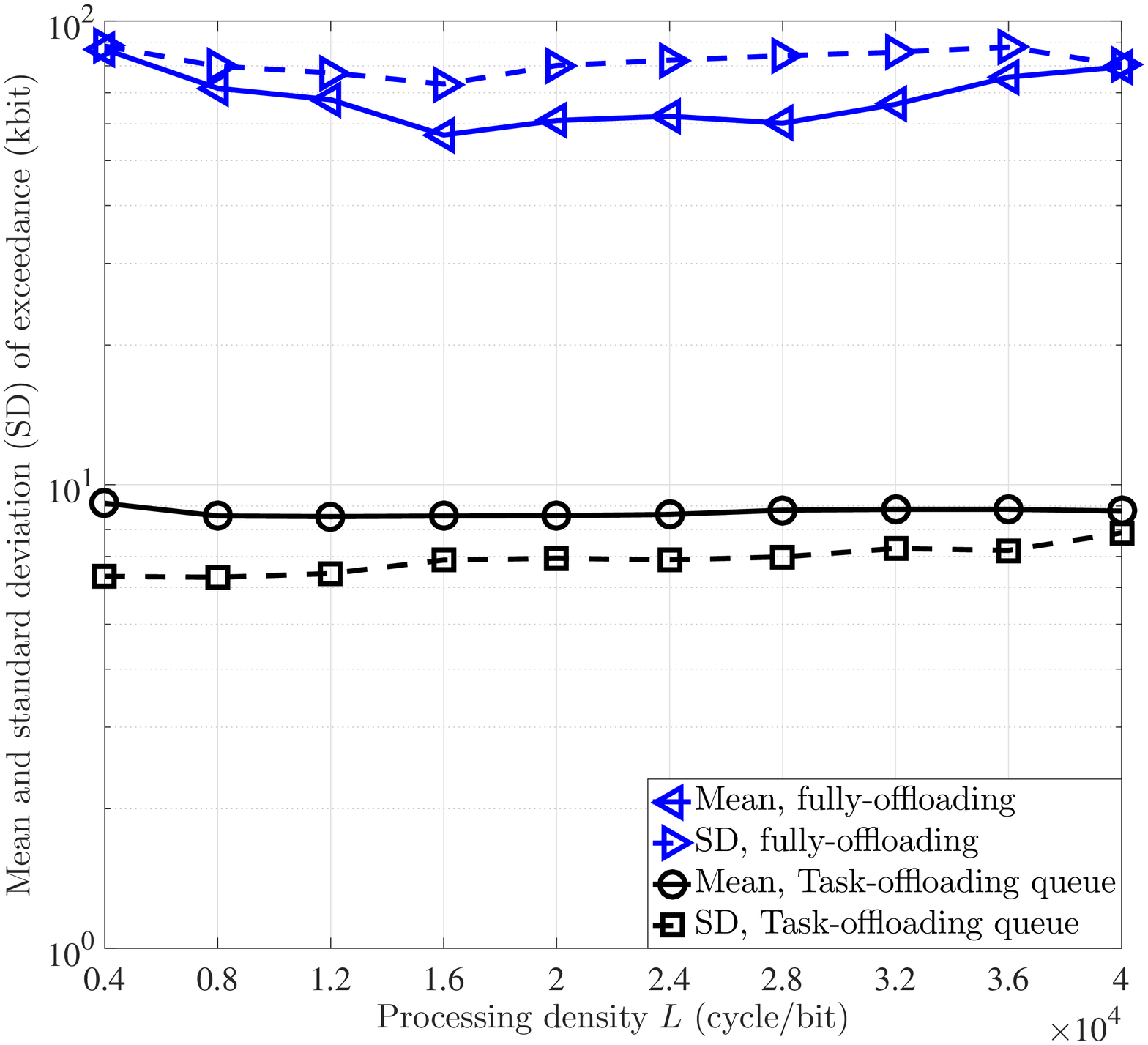}}
	\caption{1) Average end-to-end delay,  2) 99th percentile of the UE's queue length, and 3) mean and standard deviation of exceedances over the 99th percentile queue length, versus  processing density with $\lambda=100$\,kbps.}
		\label{Fig:8}
\end{figure*}

Varying the Lyapunov tradeoff parameter $V,$  we show the corresponding average power consumption in Fig.~\ref{Fig:5} and further break down the power cost in Fig.~\ref{Fig:6}.
As  realized from the objective function of problem {\bf P1-2'}, total power consumption is reduced with increasing $V$  in all network settings, and the minimal power cost can be asymptotically approached.
For any given total power consumption of Fig.~\ref{Fig:5}, we investigate, in Fig.~\ref{Fig:6}, the corresponding task split ratio and the average power consumed by local computation.
When less power is consumed, more tasks are assigned to the task-offloading queue. In other words, if the UE is equipped with weak computation capability or has lower power budget, offloading most tasks (i.e., increasing $\zeta$) helps to meet the URLLC requirements. 
Given that the same fraction of arrival tasks with a specific  processing density $L$ is computed locally, the same portion of power will be consumed in local computation regardless of the arrival rate $\lambda$. In this regard, as shown in Fig.~\ref{Fig:6}(a),  the power consumption gap between local computation and total  consumption is around 5\,dB  at $\zeta=10$  for different values of $\lambda$.
However, computation-intense tasks require higher CPU cycles in local computation. Comparing the curves in Figs.~\ref{Fig:6}(a) and \ref{Fig:6}(b), we can find that the gap is smaller with a larger $L$.
Moreover, since higher rates and more intense processing densities demand more resources for task execution, the UE consumes more power when these two parameters $L$ and $\lambda$ increase as expected.
Additionally, for the specific values of $L$, $\lambda$ and $\zeta$, the locally-executed task amounts and required computational resources in local computation are identical for different numbers of UEs.
Thus,  as shown in Fig.~\ref{Fig:6}(b), the local power consumption is the same regardless of the network density, but each UE consumes more transmit power (and total consumption) due to stronger transmission interference in the denser network.

In addition to the discussed MEC architecture of this work, we consider another two baselines for performance comparison:
\emph{(i)} a baseline with
 no  MEC servers for offloading, and \emph{(ii)} a baseline that
 offloads all the tasks to the MEC server due to the absence of  local computation capability.  In the following part, we compare the performance of the proposed method with these two baselines  for various task arrival rates in Fig.~\ref{Fig:7} and various processing densities in Fig.~\ref{Fig:8}.
We first compare the average end-to-end delay in Figs.~\ref{Fig:7}(a) and \ref{Fig:8}(a). 
For the very small arrival rate and processing density requirement, the UE's computation capability is sufficient to execute tasks rapidly, whereas transmission delay and the extra queuing delay incurred at the server degrade the performance in the fully-offloading  and partially-offloading schemes. 
\begin{figure}[t]
\centering
\subfigure[Four MEC servers with $\{2,4,8,16\}$ CPU cores.]{\includegraphics[width=\columnwidth]{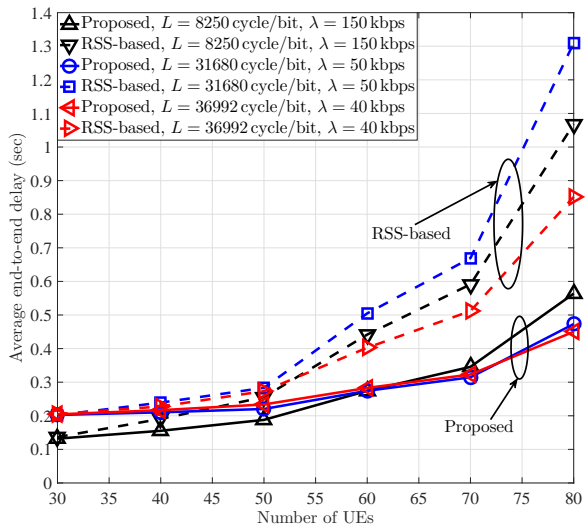}}
\subfigure[Four MEC servers with $\{8,8,8,8\}$ CPU cores.]{\includegraphics[width=\columnwidth]{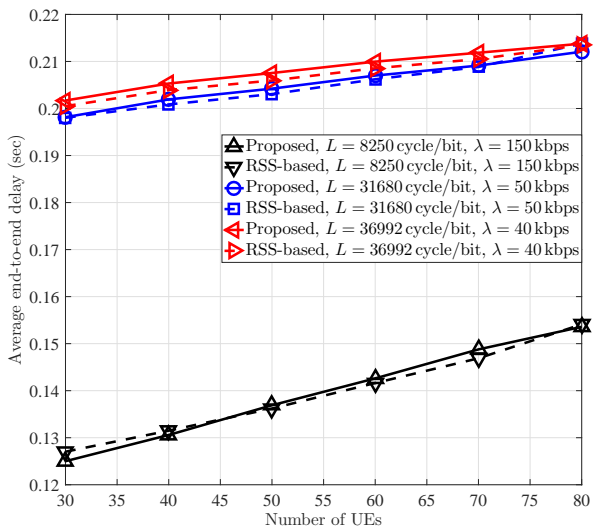}}
\caption{Average end-to-end delay versus number of UEs for different UE-server association schemes.}
	\label{Fig:9}
\end{figure}
\begin{figure}[t]
\centering
\subfigure[Four MEC servers with $\{2,4,8,16\}$ CPU cores.]{\includegraphics[width=\columnwidth]{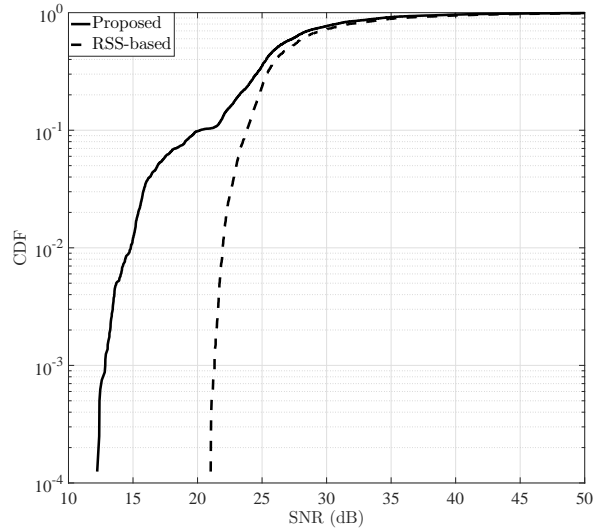}}
\subfigure[Four MEC servers with $\{8,8,8,8\}$ CPU cores.]{\includegraphics[width=\columnwidth]{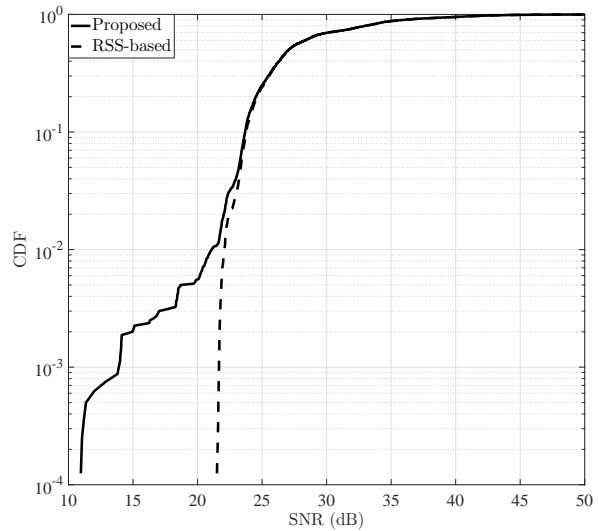}}
\caption{CDF of the wireless channel SNR between the UE and the associated MEC server for different UE-server association schemes with $L=8250$\,cycle/bit, $\lambda=150$\,kbps, and 80 UEs.}
\label{Fig:10}
\end{figure}
While increasing $L$ or $\lambda$, the UE's computation capability becomes less supportable as mentioned previously. The two (fully- and partially-) offloading schemes will eventually provide better delay performance in various network settings.
Additionally,  for the offloading schemes, the average end-to-end delay is larger in the denser network due to stronger interference and longer waiting time for the server's computational resources. Compared with the fully-offloading schemes, our approach has a remarkable delay performance improvement in the denser network since the fully-offloaded tasks incur tremendous queuing delay.
In addition to the average end-to-end delay, we also show the 99th percentile of the UE's  queue length, i.e., $q_{99}\coloneqq F^{-1}_Q(0.99)$, as a reliability measure in Figs.~\ref{Fig:7}(b), \ref{Fig:7}(c), \ref{Fig:8}(b), and \ref{Fig:8}(c). 
When the proposed approach outperforms the non-MEC scheme in terms of the average end-to-end delay, the local-computation queue  has a lower 99th percentile than the 99th percentile queue length of the non-MEC scheme. Similar results can be found for the task-offloading queue and the queue of the fully-offloading scheme if our proposed approach has a lower average end-to-end delay. In the 30-UE case, although our proposed approach has a higher average end-to-end delay than the fully-offloading scheme, splitting the tasks decreases the loading of each queue  buffer and  results the lower queue length.
Let us zoom in on the excess queue value over the  99th percentile in the 30-UE case, the mean and standard deviation of exceedances, i.e., $\mathbb{E}[Q-q_{99}|Q>q_{99}]$ and $\sqrt{\mbox{Var}(Q-q_{99}|Q>q_{99})}$, are investigated  in Figs.~\ref{Fig:7}(d) and \ref{Fig:8}(d), where we can find that our approach has a smaller amount and more concentrated extent of the extreme events.
Although fully offloading tasks achieves lower delay in the sparse network, the partially-offloading scheme can lower the loading in the queue buffer. In practice, if the queue buffer is equipped with the finite-size storage, our proposed approach can  properly address the potential overflowing issue and achieve more reliable task computation.

Finally in Figs.~\ref{Fig:9} and \ref{Fig:10}, we show the advantage of our proposed UE-server association approach in the heterogeneous MEC architecture. As a  baseline, we consider the mechanism in which the UE accesses the MEC server with the highest RSS. Provided that the deployed MEC servers have different computation capabilities (in terms of CPU cores),
associating some UEs with the stronger-capability server but the lower RSS can properly balance the servers' loadings although the transmission rates are sacrificed. As a result, compared with the baseline, the waiting time for the server's computational resources and the end-to-end delay of the proposed association approach are alleviated. The advantage is more prominent in the dense network since there are more UEs waiting for the server's resources.
If the servers' computation capabilities are identical, associating the UE with the server with the lower RSS does not give a computation gain. 
Thus, the proposed approach and baseline  have  identical association outcome and delay performance irrespective of the network setting.
We further show the CDF of the wireless channel signal-to-noise ratio  (SNR), measured as $\frac{P_{i}^{\max}\mathbb{E}[h_{ij}]}{N_0W}$, between the  UE and its associated server in Fig.~\ref{Fig:10}. When the MEC servers are homogeneous, the proposed approach and baseline have  similar association results as mentioned above. This is verified by the 1\,dB gap at the 1st percentile of the SNR. Nevertheless, in the heterogeneous MEC architecture, the gap is 7\,dB at the 1st percentile of the SNR and reduces to 1\,dB until the 30th percentile.

\section{Conclusions}\label{Sec: conclude}

The goal of this work is to propose a new low-latency and reliable communication-computing system design for enabling mission-critical applications.
In this regard, the URLLC requirement has been formulated with respect to the threshold deviation probability of the task queue length. By leveraging extreme value theory, we have characterized the statistics of the threshold deviation event with a low occurrence probability and imposed another URLLC constraint on the high-oder statistics.
The studied problem has been cast as UEs' computation and communication power minimization subject to the URLLC constraints.
Furthermore, incorporating  techniques from Lyapunov stochastic optimization and matching theory, we have proposed a two-timescale framework for UE-server association, task offloading, and resource allocation.  UE-server association is formulated as a many-to-one matching game with externalities which is addressed, in the long timescale, via the notion of swap matching. 
In every time slot, each UE allocates the computation and communication resources, and splits the task arrivals for local computation and offloading. In parallel, each server schedules its multiple CPU cores to compute the UEs' offloaded tasks. Numerical results have shown  the effectiveness of characterizing the extreme queue length using extreme value theory. Partially offloading tasks provides more reliable task computation in contrast with the non-MEC and fully-offloading schemes.
When the servers have different computation capabilities, the proposed UE-server association approach achieves lower delay than the RSS-based baseline, particularly in denser networks.
 In our future work, we will investigate distributed machine learning techniques \cite{ParkAI} to further reduce latency and reliability.
\appendices

\section{Proof of Lemma \ref{lemma power}}
\label{Lem: opt pow}

We first express the Lagrangian of problem {\bf P1-2'} as
\begin{align}\label{Eq: Lagrangian}
&\big(\beta_{j^{*}i}(t)-\beta_{i}^{\rm O}(t)\big)\tau W\times\mathbb{E}_{\hat{I}_{ij^{*}}}\bigg[ \log_2\bigg(1+\frac{P_{i}(t)h_{ij^{*}}(t)}{N_0W+I_{ij^{*}}(t)}\bigg)\bigg]\notag
\\&\quad-  \frac{\beta_{i}^{\rm L}(t)\tau f_{i}(t)}{L_i}+\gamma\big( \kappa [f_i(t)]^3+P_{i}(t) -P_{i}^{\max}\big)\notag
 \\ &\qquad+V\big(\kappa [f_i(t)]^3+P_{i}(t)\big)-\alpha_1 f_{i}(t)-\alpha_2 P_{i}(t),
\end{align}
where $\gamma$, $\alpha_1$, and $\alpha_2$ are the Lagrange multipliers.
Taking the partial differentiations of \eqref{Eq: Lagrangian} with respect to $f_i$ and $P_i$, we have
\begin{align*}
\frac{\partial \eqref{Eq: Lagrangian}}{\partial f_i}&= -\frac{\beta_{i}^{\rm L}(t)\tau }{L_i}+3\kappa (V+\gamma) [f_i(t)]^2-\alpha_1,
 %
\\\frac{\partial \eqref{Eq: Lagrangian}}{\partial P_i}&=
  \mathbb{E}_{I_{ij^{*}}}\bigg[\frac{\big(\beta_{j^{*}i}(t)-\beta_{i}^{\rm O}(t)\big)\tau Wh_{ij^{*}}}{(N_0W+I_{ij^{*}}+P_{i}h_{ij^{*}})\ln 2}\bigg] +V+\gamma-\alpha_2.
%
\end{align*}
Subsequently, since  {\bf P1-2'} is  a convex optimization problem, we apply the Karush--Kuhn--Tucker (KKT) conditions to derive the optimal solution in which the optimal CPU-cycle $f^{*}_i(t)$, optimal transmit power $P_i^*(t)$, and optimal Lagrange multipliers (i.e., $\gamma^*$, $\alpha_1^*$, and $\alpha_2^*$) satisfy
\begin{align}
 & f^{*}_i(t)=\sqrt{\frac{\alpha^*_1 + \frac{\beta_{i}^{\rm L}(t)\tau }{L_i}}{ 3\kappa(V+\gamma^*)}},\label{Eq: KKT-1}
\\&  \mathbb{E}_{I_{ij^{*}}}\bigg[\frac{\big(\beta_{i}^{\rm O}(t)-\beta_{j^{*}i}(t)\big)\tau W h_{ij^{*}}}{(N_0W+I_{ij^{*}}+P^{*}_{i}(t)h_{ij^{*}}(t))\ln 2}\bigg] =V +\gamma^*-\alpha_2^*,\label{Eq: KKT-2}
%
\\&f^*_{i}(t)\geq 0,
\qquad\alpha_1^*\geq 0,
\qquad\alpha_1^*f^*_{i}(t)= 0,\label{Eq: KKT-3}
\\&P^*_{i}(t)\geq 0,
\qquad\alpha^*_2\geq 0,
\qquad\alpha^*_2P_{i}(t)= 0,\label{Eq: KKT-4}
\end{align}
and
\begin{align}\label{Eq: KKT-5}
\begin{cases}
 \kappa [f^*_i(t)]^3+P^*_{i}(t) -P_{i}^{\max}\leq 0,
\\\gamma^*\geq 0,
\\\gamma\big(\kappa [f^*_i(t)]^3+P^*_{i}(t) -P_{i}^{\max}\big)= 0.
\end{cases}
\end{align}
From \eqref{Eq: KKT-1} and \eqref{Eq: KKT-3}, we deduce 
\begin{align*}
f^{*}_i(t)=
\sqrt{\frac{
 \beta_{i}^{\rm L}(t)\tau }{ 3L_i\kappa(V+\gamma^{*})}}.
 \end{align*}
Additionally,  from \eqref{Eq: KKT-2} and \eqref{Eq: KKT-4}, we can find that if 
$$  \mathbb{E}_{I_{ij^{*}}}\Big[\frac{(\beta_{i}^{\rm O}(t)-\beta_{j^{*}i}(t))\tau W h_{ij^{*}}}{(N_0W+I_{ij^{*}})\ln 2}\Big]>V+\gamma^{*},$$
 we have a positive optimal transmit  power $P^{*}_{ij}>0$ which satisfies
\begin{align}\notag
  \mathbb{E}_{I_{ij^{*}}}\bigg[\frac{\big(\beta_{i}^{\rm O}(t)-\beta_{j^{*}i}(t)\big)\tau W h_{ij^{*}}}{(N_0W+I_{ij^{*}}+P^{*}_{i}h_{ij^{*}})\ln 2}\bigg]
 =V+\gamma^{*}.
\end{align}
Otherwise,  $P^{*}_{i}=0$. 
Moreover, note that $\gamma^{*}$ is 0 if $\kappa [f^{*}_i(t)]^3+P^{*}_{i}(t)< P_{i}^{\max}$. When $\gamma^{*}> 0$,
$\kappa [f^{*}_i(t)]^3+P^{*}_{i}(t)= P_{i}^{\max}$.

\bibliographystyle{IEEEtran}
\bibliography{reffog}

\end{document}